\documentclass[10pt]{amsart}
\usepackage{amsmath,amssymb,amsthm,graphicx,mathrsfs,url}
\usepackage[usenames,dvipsnames]{color}
\definecolor{darkred}{rgb}{0.4,0.1,0.1}
\usepackage[colorlinks=true,linkcolor =darkred, citecolor=Blue]{hyperref}
\usepackage{rotating}
\usepackage{float}
\usepackage{caption}
\usepackage{epic}
\usepackage{dashrule}

\def\softness{0.4}
\definecolor{softred}{rgb}{1,\softness,\softness}
\definecolor{softgreen}{rgb}{\softness,1,\softness}
\definecolor{softblue}{rgb}{\softness,\softness,1}

\definecolor{softrg}{rgb}{1,1,\softness}
\definecolor{softrb}{rgb}{1,\softness,1}
\definecolor{softgb}{rgb}{\softness,1,1}

\usepackage[latin1]{inputenc}

\usepackage[T1]{fontenc}

\usepackage{graphics}
\numberwithin{figure}{section}
\numberwithin{equation}{section}
\theoremstyle{plain}
\newtheorem{thm}{Theorem}[section]
\newtheorem{lem}[thm]{Lemma}

\newtheorem{prop}[thm]{Proposition}
\newtheorem{example}[thm]{Example}

\newtheorem{cor}[thm]{Corollary}
\newtheorem{hyp}[thm]{Hypothesis}

\newtheorem{dfn}[thm]{Definition}
\theoremstyle{remark}
\newtheorem{remark}[thm]{Remark}
\theoremstyle{plain}

\newcommand{\supp}{\mathrm{supp}\,}

\newcommand{\be}{\begin{equation}}
\newcommand{\ee}{\end{equation}}
\newcommand{\beu}{\begin{equation*}}
\newcommand{\eeu}{\end{equation*}}
\newcommand{\besu}{\begin{equation*}
\begin{aligned}}
\newcommand{\eesu}{\end{aligned}
\end{equation*}}
\newcommand{\bes}{\begin{equation}
\begin{aligned}}
\newcommand{\ees}{\end{aligned}
\end{equation}}

\newcommand\cB{\mathcal B}

\newcommand\cG{\mathcal G}
\newcommand\cH{\mathcal H}

\newcommand\cP{\mathcal P}
\newcommand\cR{\mathcal R}
\newcommand\cS{\mathcal S}

\newcommand\fra{\mathfrak a}

\newcommand\ov{\overline}
\newcommand\wt{\widetilde}

\newcommand\sess{\sigma_{\rm ess}}

\newcommand\sign{{\rm sign\,}}

\newcommand\void[1]{}

\def\sess{\sigma_{\rm ess}}
\def\ran{{\rm ran\,}}

      \def\dC{{\mathbb C}}

   \def\dN{{\mathbb N}}   
      \def\dR{{\mathbb R}}

   \def\cB{{\mathcal B}}   
      
\def\cG{{\mathcal G}}   \def\cH{{\mathcal H}}

\def\cP{{\mathcal P}}      \def\cR{{\mathcal R}}
\def\cS{{\mathcal S}}      
      
   \def\cZ{{\mathcal Z}}

\newcommand{\dom}{\mathrm{dom}\,}

\title[$\delta$ and $\delta'$-interactions on Lipschitz partitions]{Schr\"odinger operators with $\delta$ and $\delta'$-interactions 
on Lipschitz surfaces and chromatic numbers of associated partitions}

\author{Jussi Behrndt \and Pavel Exner \and Vladimir Lotoreichik}

\newcommand{\Real}{\mbox{{\rm Re}}\,}

\makeatletter
\def\@tocline#1#2#3#4#5#6#7{\relax
  \ifnum #1>\c@tocdepth 
  \else
    \par \addpenalty\@secpenalty\addvspace{#2}%
    \begingroup \hyphenpenalty\@M
    \@ifempty{#4}{%
      \@tempdima\csname r@tocindent\number#1\endcsname\relax
    }{%
      \@tempdima#4\relax
    }%
    \parindent\z@ \leftskip#3\relax \advance\leftskip\@tempdima\relax
    \rightskip\@pnumwidth plus4em \parfillskip-\@pnumwidth
    #5\leavevmode\hskip-\@tempdima
      \ifcase #1
       \or\or \hskip 1em \or \hskip 2em \else \hskip 3em \fi%
      #6\nobreak\relax
    \dotfill\hbox to\@pnumwidth{\@tocpagenum{#7}}\par
    \nobreak
    \endgroup
  \fi}
\makeatother
\begin{document}

\begin{abstract}
We investigate Schr\"odinger operators with $\delta$ and $\delta'$-interactions
supported on hypersurfaces, which separate the Euc\-lidean space into finitely many  bounded and unbounded Lipschitz domains. 
It turns out that the combinatorial properties of the partition and the spectral properties
of the corresponding operators are related. 
As the main result we prove an operator inequality for the Schr\"odinger operators with $\delta$ and $\delta'$-interactions
which  is based on an optimal colouring and involves the chromatic number of the partition. 
This inequality implies various relations for the spectra of the Schr\"odinger operators 
and, in particular, it allows to transform known results for  Schr\"odinger operators
with $\delta$-interactions to Schr\"odinger operators with $\delta^\prime$-inter\-actions.

\end{abstract}

\maketitle

\section{Introduction}

Schr\"o\-dinger operators with singular $\delta$-type interactions supported on discrete sets, curves and 
surfaces are used for the description
of quantum mechanical systems with a certain degree of idealization. 
The spectral properties of 
Schr\"o\-dinger operators with $\delta$ and $\delta^\prime$-interactions were investigated in
numerous mathematical and physical articles in the recent past; we mention only \cite{BN11, KM10, MS12, O10} for interactions on point sets, 
\cite{CK11, EI01, EK08, EN03, EP12, K12, KV07} on curves, and \cite{AKMN13, BLL13, EF09,  EK03} for interactions on surfaces.
For a survey and further references we refer the reader to \cite{E08} and to the standard monograph \cite{AGHH}.

In this paper we investigate attractive $\delta$ and $\delta'$-interactions 
supported on general hypersurfaces, which separate the Euclidean space $\dR^d$ into finitely many bounded and unbounded Lipschitz domains. 
We establish a connection between the combinatorial properties of these so-called Lipschitz partitions and the relation of the
Schr\"odinger operators with $\delta$ and $\delta'$-interactions to each other. More precisely, suppose that the Euclidean space $\dR^d$, $d\geq 2$, 
is split into a finite number of Lipschitz domains $\Omega_k$, $k=1,\dots,n$, and let $\Sigma$ be the union of the boundaries of all $\Omega_k$.
The chromatic number $\chi$ of the partition is defined as the minimal number of colours, 
which is sufficient to colour all domains $\Omega_k$ in such a way that any two neighbouring domains have distinct colours. 
In the two dimensional case the famous four colour theorem states that $\chi\le 4$ for any Lipschitz partition of the plane.
In the following the strengths of the $\delta$ and $\delta'$-interactions are assumed to be constant along their support $\Sigma$, 
which simplifies the explanation of our results. 
Let $\alpha \in \dR$, $\beta\in\dR\setminus\{0\}$ and define the quadratic forms 
\[
\fra_{\delta,\alpha}[f] := \big\|\nabla f\big\|_{L^2(\dR^d;\dC^d)}^2 - \alpha\big\| f|_\Sigma\big\|_{L^2(\Sigma)}^2,\quad \dom \fra_{\delta,\alpha} = H^1(\dR^d),
\]
and
\[
\begin{split}
&\fra_{\delta',\beta}[f] := \sum_{k=1}^n\big\|\nabla f_k\big\|^2_{L^2(\Omega_k;\dC^d)} -
 \sum_{k=1}^{n-1}\sum_{l=k+1}^n \beta^{-1}\big\|f_k|_{\Sigma_{kl}} - f_l|_{\Sigma_{kl}}\big\|_{L^2(\Sigma_{kl})}^2,\\
&\dom\fra_{\delta',\beta} = \bigoplus_{k=1}^n H^1(\Omega_k),
\end{split}
\]
where $f_k = f|_{\Omega_k}$ and $\Sigma_{kl}=\partial\Omega_k\cap\partial\Omega_l$, $k\not =l$.
It turns out that $\fra_{\delta,\alpha}$ and $\fra_{\delta',\beta}$ are densely defined, closed, symmetric
forms in the Hilbert space $L^2(\dR^d)$ which are semibounded from below, and hence $\fra_{\delta,\alpha}$ and $\fra_{\delta',\beta}$ 
induce self-adjoint operators 
$-\Delta_{\delta,\alpha}$ and $-\Delta_{\delta',\beta}$ in  $L^2(\dR^d)$. It will be shown in Theorem~\ref{thm:delta} that these
operators act as minus Laplacians and the functions in their domains satisfy appropriate $\delta$ and $\delta'$-boundary conditions on $\Sigma$.

Our main result, Theorem~\ref{thm:main}, is an inequality for the quadratic forms $\fra_{\delta,\alpha}$ and $\fra_{\delta',\beta}$, or equivalently, for the 
Schr\"odinger operators $-\Delta_{\delta,\alpha}$ and $-\Delta_{\delta',\beta}$ with  $\delta$-interaction of strength $\alpha$ and  $\delta'$-interaction 
of strength $\beta$, respectively. Namely, if $\alpha$, $\beta$ and the chromatic number $\chi$ of the partition satisfy 
\begin{equation}
\label{assumption}
0 < \beta \le \frac{4}{\alpha}\sin^2\big(\pi / \chi\big)
\end{equation}
then it will be shown that there exists an unitary operator $U$ in $L^2(\dR^d)$ such that
\begin{equation}
\label{inequality}
U^{-1}(-\Delta_{\delta',\beta})U\le -\Delta_{\delta,\alpha}
\end{equation}
holds. The operator $U$ can be constructed explicitly
as soon as the optimal colouring of the partition is provided. The value $4\sin^2(\pi / \chi)$ in \eqref{assumption} pops up
as the square of the edge length of the equilateral polygon with $\chi$ vertices, which is circumscribed in the unit circle on the complex plane. 
We also discuss the sharpness of Theorem~\ref{thm:main}  for some cases. First of all it is shown in Example \ref{exchi2} that the assumption 
\eqref{assumption}
is sharp if $\chi=2$. 
In Section~\ref{sec:chi3} we then discuss the case $\chi=3$.
It turns out that the weaker assumption $0 < \beta \le \tfrac{4}{\alpha}$
(corresponding to $\chi=2$ in \eqref{assumption}) is not 
sufficient for the existence of a unitary operator $U$ such that \eqref{inequality} holds for every partition with $\chi = 3$.
This fact will be shown explicitly by considering
a symmetric star-graph with three leads as the support of the $\delta$ and $\delta^\prime$-interaction.

The inequality \eqref{inequality} is particularly useful since it implies various relations of the spectra of $-\Delta_{\delta,\alpha}$ and $-\Delta_{\delta',\beta}$,
and it allows to transform known results for  Schr\"odinger operators
with $\delta$-interactions to Schr\"odinger operators with $\delta^\prime$-inter\-actions. 
We apply our main theorem and its consequences to Lipschitz partitions with compact boundary and so-called locally deformed
partitions, where also unbounded Lipschitz domains with unbounded boundaries appear. In these situations we are able to determine or to describe
the essential spectra of $-\Delta_{\delta,\alpha}$ and $-\Delta_{\delta',\beta}$, and we derive some consequences on the spectral properties of
$-\Delta_{\delta',\beta}$. In particular, it turns out that $-\Delta_{\delta',\beta}$
has a non-empty discrete spectrum if the same holds for $-\Delta_{\delta,\alpha}$, and hence we conclude results on the existence of 
deformation-induced bound states of $-\Delta_{\delta',\beta}$ from the corresponding results in \cite{EI01,EK03} for the $\delta$-case. 
We mention that various results on the spectral properties of Schr\"odinger operators with
$\delta$-interactions supported by locally deformed or weakly  straight lines and hyperplanes or under more general assumptions of asymptotic flatness exist in the mathematical literature,
see, e.g. \cite{CK11, EI01, EK03, EK05, LLP10}.

The structure of the paper is as follows. In Section~\ref{sec:prelim} some preliminary facts on the ordering of quadratic forms, Lipschitz partitions,
and  Sobolev spaces on Lipschitz domains are provided. The quadratic forms $\fra_{\delta,\alpha}$ and $\fra_{\delta',\beta}$, and the
corresponding Schr\"{o}dinger operators $-\Delta_{\delta,\alpha}$ and $-\Delta_{\delta',\beta}$ are introduced and studied in Section~\ref{sec:main}.
This section contains also the main result, Theorem~\ref{thm:main}, and some examples. The more computational aspects in 
the example of a symmetric star graph with three leads were outsourced in an appendix. The essential spectra and bound states 
of Schr\"odinger operators with $\delta$ and $\delta'$-interactions on Lipschitz partitions with compact boundary 
and locally deformed Lipschitz partitions
are studied in Section~\ref{sec4}.

\subsection*{Acknowledgements}

The authors gratefully acknowledge financial support by the Austrian 
Science Fund (FWF), project
P 25162-N26, Czech Science Foundation (GA\v{C}R), project P203/11/0701, 
and the Austria-Czech Republic cooperation grant  CZ01/2013.

\section{Preliminaries}\label{sec:prelim}

In this paper we use mainly standard facts from operator theory in Hilbert spaces and 
basic properties of Sobolev spaces on Lipschitz domains.
In this section we briefly recall and define some notions on semibounded sesquilinear forms, Lipschitz partitions and Sobolev spaces.

\subsection{Ordering of sesquilinear forms}\label{ssec:forms}

The self-adjoint operators in this paper are introduced with the help of closed, densely defined, semibounded, symmetric sesquilinear forms 
via the first representation theorem \cite[VI Theorem~2.1]{Kato}. 
For a comprehensive introduction into the theory of forms we refer the reader to \cite[Chapter VI]{Kato}, \cite[Chapter 10]{BS87},
and \cite[Chapter 4.6]{BEH08}. 


First we recall the ordering of forms and associated self-adjoint operators.

\begin{dfn}\label{dfn:forms}
Let $\fra_1$ and $\fra_2$ be closed, densely defined, symmetric sesquilinear forms in a Hilbert space $\cH$ and assume that 
$\fra_1$ and $\fra_2$ are bounded from below. Then we shall write $\fra_2\leq \fra_1$ if
\[
\dom\fra_1\subset\dom\fra_2 \quad\text{and}\quad \fra_2[f]\le\fra_1[f]\quad\text{for all}\quad f\in\dom\fra_1.
\]
If $H_1$ and $H_2$ denote the self-adjoint operators associated with $\fra_1$ and $\fra_2$ in $\cH$, respectively, then we write $H_2\leq H_1$
if and only if $\fra_2\leq \fra_1$.
\end{dfn}

We note that by \cite[VI Theorem 2.21]{Kato} two self-adjoint operators $H_1$ and $H_2$ which are 
semibounded from below by $\nu_1$ and $\nu_2$, respectively, 
satisfy $H_2\leq H_1$ if and only if for some, and hence for all, $\nu<\min\{\nu_1,\nu_2\}$ 
\begin{equation*}
(H_2 -\nu)^{-1} - (H_1 - \nu)^{-1} \ge 0.
\end{equation*}

The essential spectrum  of a self-adjoint operator $H$ is denoted by $\sess(H)$. If $\sess(H)=\varnothing$ we
set $\min\sess(H)=+\infty$ in the following definition.

\begin{dfn}\label{dfn:spec1}
Let $H$ be a self-adjoint operator in an infinite dimensional Hilbert space and assume that $H$ is bounded from below.
We set 
\[
N(H) := \sharp \big\{\lambda\in (-\infty,\min\sess(H)) :\lambda \in\sigma_{\rm p}(H)\big\}\in\dN_0\cup\{\infty\}
\]
and
denote by $\{\lambda_k(H)\}_{k=1}^\infty$ the sequence of eigenvalues of $H$ lying below $\min\sess(H)$, 
enumerated in non-decreasing order and repeated with multiplicities. In the case $N(H) <\infty$
this sequence is extended by setting $\lambda_{N(H)+k}(H)=\min\sess(H)$, $k\in\dN$.
\end{dfn}

The statements in the next theorem are consequences of the min-max principle, see, 
e.g. \cite[10.2 Theorem 4]{BS87} and \cite[Theorem XIII.2]{RS78}.

\begin{thm}\label{thm:variation}
Let $\fra_1$ and $\fra_2$ be closed, densely defined, symmetric sesquilinear forms in $\cH$
which are bounded from below and let $H_1$ and $H_2$ be the corresponding self-adjoint operators.
Assume that $\fra_2 \leq\fra_1$, or equivalenty that $H_2\leq H_1$. Let $\{\lambda_k(H_i)\}_{k=1}^\infty$ 
and $N(H_i)$, $i=1,2$, be as in Definition~\ref{dfn:spec1}.
Then the following statements hold:
\begin{itemize}\setlength{\itemsep}{1.2ex}
\item [\rm (i)] $\lambda_k(H_2)\le \lambda_k(H_1)$ for all $k\in\dN$;
\item[\rm (ii)] $\min\sess(H_2) \le \min\sess(H_1)$;
\item [\rm (iii)] If $\min\sess(H_1) = \min\sess(H_2)$ then $N(H_1) \le N(H_2)$.
\end{itemize}
\end{thm}

\subsection{Lipschitz partitions of Euclidean spaces}
\label{ssec:partition}

In this short subsection we introduce the notion of finite Lipschitz partitions and discuss a
combinatorial property of these partitions. For the definition and basic properties of
Lipschitz domains we refer the reader to \cite[VI.3]{Stein}.

\begin{dfn}
\label{def:partition}
A finite family of Lipschitz domains $\cP = \{\Omega_k\}_{k=1}^n$ is called a Lipschitz partition of $\dR^d$, $d\geq 2$, 
if $$\dR^d = \bigcup_{k=1}^n\overline\Omega_k\qquad\text{and}\qquad \Omega_k\cap\Omega_l = \varnothing, \qquad k,l=1,2,\dots,n,\,\,\,k\ne l.$$ 
The union $\cup_{k=1}^n\partial\Omega_k =: \Sigma$ is the boundary of the Lipschitz partition $\cP$.
For $k\not= l$ we set  $\Sigma_{kl}:=\partial\Omega_k\cap\partial\Omega_l$ and we say 
that $\Omega_k$ and $\Omega_l$,  $k\ne l$, 
are neighbouring domains if $\sigma_{k}(\Sigma_{kl})>0$, where $\sigma_k$ denotes the Lebesgue measure on $\partial\Omega_k$.
\end{dfn}

The chromatic number of a Lipschitz partition is defined with the help of colouring mappings.

\begin{dfn}\label{dfn:chro}
Let $\cP = \{\Omega_k\}_{k=1}^n$ be a Lipschitz partition of $\dR^d$, $d\geq 2$, with $\Sigma_{kl}=\partial\Omega_k\cap\partial\Omega_l$,
$k\not= l$.
Then a mapping $\varphi \colon \{1,2,\dots,n\}\rightarrow \{0,1,\dots, m-1\}$ is called an $m$-colouring for $\cP$ if 
\[
\sigma_k(\Sigma_{kl})>0 \quad \Longrightarrow \quad \varphi(k) \ne \varphi(l)
\] 
for all $k,l=1,2,\dots, n$, $k\ne l$.
The chromatic number $\chi$ of the Lipschitz partition $\cP$ is 
defined as
\[
\chi := \min\bigr\{ m\in\dN\colon \exists\, \text{$m$-colouring mapping for}~ \cP\bigr\}.
\]
\end{dfn}

Thus the chromatic number $\chi$ of a Lipschitz partition $\cP = \{\Omega_k\}_{k=1}^n$ of $\dR^d$ 
is the minimal
number of colours, which is  sufficient to colour all domains $\Omega_k$ such that any two neighbouring domains have different colours;
recall that $\Omega_k$ and $\Omega_l$ are regarded as neighbouring domains only if the Lebesgue measure of $\Sigma_{kl}=\partial\Omega_k\cap
\partial\Omega_l$ is positive. 
As a famous example we mention the four colour theorem which states that the chromatic
number of any Lipschitz partition $\cP$ of $\dR^2$ is $\chi\le 4$.

\subsection{Sobolev spaces on arbitrary Lipschitz domains}
\label{ssec:Sobolev1}
For a Lipschitz domain $\Omega$ we denote the standard $L^2$-based Sobolev spaces on $\Omega$ and $\partial\Omega$
by $H^s(\Omega)$, $s\in\dR$, and $H^t(\partial\Omega)$, $t\in [-1,1]$, respectively.
For the definition and general properties of Sobolev
spaces on Lipschitz domains and their boundaries we refer the reader to \cite{McLean} and \cite{Stein}.
Recall that for a Lipschitz domain $\Omega\subset\dR^d$ there exists an {\it extension operator}
\[
E\colon L^2(\Omega)\rightarrow L^2(\dR^d)
\]
satisfying the following conditions: 
\begin{itemize}\setlength{\itemsep}{1.2ex}
\item [\rm (i)] $(E f)\upharpoonright \Omega = f$ for all $f\in L^2(\Omega)$;
\item [\rm (ii)] $E (H^k(\Omega))\subset H^k(\dR^d)$ for all $k\in\dN_0$;
\item [\rm (iii)] $E:H^k(\Omega)\rightarrow  H^k(\dR^d)$ is continuous for all $k\in\dN_0$.  
\end{itemize}

The useful estimate on the trace in the next lemma is essentially a consequence of the continuity of the trace map
and the above mentioned properties of the extension operator.
For the convenience of the reader we provide a short proof.

\begin{lem}\label{lem:trace}
Let $\Omega\subset\dR^d$ be a bounded or unbounded Lipschitz domain.  Then for any $\varepsilon >0$ 
there exists a constant $C(\varepsilon) > 0$ such that
\[
\|f|_{\partial\Omega}\|_{L^2(\partial\Omega)}^2 \le \varepsilon \|\nabla f\|_{L^2(\Omega;\dC^d)}^2 + C(\varepsilon)\|f\|^2_{L^2(\Omega)}
\]
holds for all $f\in H^1(\Omega)$.
\end{lem}

\begin{proof}
Let $f\in H^1(\Omega)$, fix some $s\in (\frac12,1)$ and let $Ef\in H^1(\dR^d)$ be the extension of $f$.
The continuity of the trace \cite{M87,Necas} and the properties of the extension operator
imply that there exists $c>0$ such that
\begin{equation*}
\|f|_{\partial\Omega}\|_{L^2(\partial\Omega)}\le c\|f\|_{H^s(\Omega)}\le c\|Ef\|_{H^s(\dR^d)}.
\end{equation*}
Hence for $\varepsilon >0$ there exists a constant $C_1(\varepsilon) >0$ such that
\[
\|f|_{\partial\Omega}\|_{L^2(\partial\Omega)}\le c\|Ef\|_{H^s(\dR^d)}\leq
\varepsilon\|Ef\|_{H^1(\dR^d)} + C_1(\varepsilon)\|Ef\|_{L^2(\dR^d)},
\]
see, e.g. \cite[Theorem 3.30]{HT} or \cite[Satz 11.18 (e)]{W00}. As $E$ is continuous (see property (iii) for $k=0$ and $k=1$))
we conclude that for $\varepsilon >0$ there exists $C_2(\varepsilon) >0$ such that
\[
\|f|_{\partial\Omega}\|_{L^2(\partial\Omega)}\le \varepsilon\|f\|_{H^1(\Omega)} + C_2(\varepsilon)\|f\|_{L^2(\Omega)}.
\]
Thus for $\varepsilon >0$ there exists a constant $C_3(\varepsilon) >0$ such that
\[
\|f|_{\partial\Omega}\|_{L^2(\partial\Omega)}^2\le \varepsilon\|f\|_{H^1(\Omega)}^2 + C_3(\varepsilon)\|f\|_{L^2(\Omega)}^2
\]
and hence the assertion follows from $\|f\|_{H^1(\Omega)}^2 = \|\nabla f\|^2_{L^2(\Omega;\dC^d)} + \|f\|_{L^2(\Omega)}^2$.
\end{proof}

For our purposes it is convenient to define the Laplacian and the Neumann trace in a weak sense in $L^2$.

\begin{dfn}\label{dfn:LaplacianNeumann}
Let $\Omega$ be a Lipschitz domain and let $u \in H^1(\Omega)$. 
\begin{itemize}
 \item [{\rm (i)}] If there exists $f\in L^2(\Omega)$ such that
\[
(\nabla u, \nabla v)_{L^2(\Omega;\dC^d)} = (f, v)_{L^2(\Omega)}\quad \text{for all}~ v\in H^1_0(\Omega)
\]
then we define $-\Delta u := f$ and say that $\Delta u \in L^2(\Omega)$.
 \item [{\rm (ii)}] If $\Delta u \in L^2(\Omega)$ and there exists $b \in L^2(\partial\Omega)$ such that
\[
(\nabla u,\nabla v)_{L^2(\Omega;\dC^d)} - (-\Delta u,v)_{L^2(\Omega)} = (b, v|_{\partial\Omega})_{L^2(\partial\Omega)}
\quad \text{for all}~v \in H^1(\Omega)
\]
then we define $\partial_\nu u|_{\partial\Omega} := b$ and say that $\partial_\nu u|_{\partial\Omega} \in L^2(\partial\Omega)$.
\end{itemize}
\end{dfn}

We note that $\Delta u$ and $\partial_\nu u|_{\partial\Omega} $ in the above definition (if they exist) are unique since 
$H^1_0(\Omega)$ is dense in $L^2(\Omega)$ and
the space $\{v|_{\partial\Omega}\colon v\in H^1(\Omega)\}$ is dense in $L^2(\partial\Omega)$, respectively; 
cf.  \cite[Theorem 3.37]{McLean}. 

\begin{dfn}\label{dfn:deltatrace}
Let $\cP = \{\Omega_k\}_{k=1}^n$ be a Lipschitz partition of $\dR^d$, $d\geq 2$, with boundary $\Sigma$. 
Let $u \in H^1(\dR^d)$, denote the restrictions
of $u$ onto $\Omega_k$ by $u_k$
and assume that 
$\Delta u_k \in L^2(\Omega_k)$ for all $k=1,2,\dots,n$. If there exists $b \in L^2(\Sigma)$ such that
\[
\big(\nabla u,\nabla v\big)_{L^2(\dR^d;\dC^d)} - \big(\oplus_{k=1}^n(-\Delta u_k),v\big)_{L^2(\dR^d)} = 
(b, v|_{\Sigma})_{L^2(\Sigma)}\!\quad \text{for all}~v \in H^1(\dR^d)
\]
then we define $\partial_\cP u|_\Sigma := b$ and say that $\partial_\cP u|_\Sigma  \in L^2(\Sigma)$. 
\end{dfn}

Let $\cP = \{\Omega_k\}_{k=1}^n$ be a Lipschitz partition with boundary $\Sigma$ and let $u\in H^1(\dR^d)$. 
As $\{v|_{\Sigma}\colon v\in H^1(\dR^d)\}$ is dense in $L^2(\Sigma)$ it follows that
$\partial_\cP u|_\Sigma$ (if it exists) is unique. 

\begin{remark}\label{remdeltalocal}
Let $\cP = \{\Omega_k\}_{k=1}^n$ be a Lipschitz partition of $\dR^d$ and assume that $\partial_\cP u|_\Sigma  \in L^2(\Sigma)$ exists for some 
$u\in H^1(\dR^d)$ in the sense of Definition~\ref{dfn:deltatrace}. Let $\Omega_k$ and $\Omega_l$ be neighbouring domains and assume
that the Neumann traces $\partial_{\nu_k}u_k\vert_{\partial\Omega_k}\in L^2(\partial\Omega_k)$ and 
$\partial_{\nu_l}u_l\vert_{\partial\Omega_l}\in L^2(\partial\Omega_l)$ exist in the sense of Definition~\ref{dfn:LaplacianNeumann}~(ii).
Let $\Gamma$ be a bounded open subset of $\Sigma_{kl}$ which is part of a Lipschitz dissection in the sense of \cite[page 99]{McLean}
such that $\overline\Gamma\cap\Sigma_{km}=\varnothing$ for $m=1,\dots,n$ with $m\not=l,k$. Then it follows that
\begin{equation}\label{deltalocal}
 \partial_\cP u\vert_\Gamma=\partial_{\nu_k}u_k\vert_\Gamma+\partial_{\nu_l}u_l\vert_\Gamma
\end{equation}
holds. In particular, if $\cP = \{\Omega_1,\Omega_2\}$ with 
$\Omega_2=\dR^d\setminus\overline\Omega_1$, boundary $\Sigma=\partial\Omega$ and the Neumann
traces exist then
\begin{equation*}
 \partial_\cP u|_{\Sigma}=\partial_{\nu_1}u_1\vert_{\Sigma}+\partial_{\nu_2}
 u_2\vert_{\Sigma}.
\end{equation*}
\end{remark}

\section{Schr\"odinger operators with $\delta$ and $\delta'$-interactions associated with Lipschitz partitions}
\label{sec:main}

In this section we define and study self-adjoint Schr\"odinger operators with $\delta$ and $\delta'$-interac\-tions 
supported on the boundary $\Sigma$ of a Lipschitz partition $\cP = \{\Omega_k\}_{k=1}^n$ of $\dR^d$, $d \ge 2$.
As the main result we prove an operator inequality 
between the $\delta$ and $\delta'$-operator, which implies a certain ordering of their spectra.  
The key assumption for this inequality is expressed in terms of the chromatic number of the Lipschitz partition.

\subsection{Free and Neumann Laplacians}

Let in the following $\cP = \{\Omega_k\}_{k=1}^n$ be a Lipschitz partition of $\dR^d$ 
with the boundary $\Sigma$. The functions $f\in L^2(\dR^d)$ will be decomposed in the form
\[
f = \oplus_{k=1}^n f_k,\qquad f_k := f\vert_{\Omega_k}\in L^2(\Omega_k),\qquad\,k =1,2,\dots,n.
\]

The {\it free Laplacian} $-\Delta_{\rm free}$ and the {\it Neumann Laplacian} $-\Delta_{\rm N}$ with Neumann boundary conditions on $\Sigma$ 
are defined as the self-adjoint operators in $L^2(\dR^d)$ associated with the sesquilinear forms 
\begin{equation}
\label{freeNform}
\begin{split}
\fra_{\rm free}[f,g] := \big(\nabla f, \nabla g\big)_{L^2(\dR^d;\dC^d)},&\qquad \dom \fra_{\rm free} := H^1(\dR^d),\\
\fra_{\rm N}[f,g]  := \sum_{k=1}^n\big(\nabla f_k, \nabla g_k\big)_{L^2(\Omega_k;\dC^d)},&\qquad \dom \fra_{\rm N} := 
\bigoplus_{k=1}^n H^1(\Omega_k),
\end{split}
\end{equation}
which  are symmetric, closed and semibounded from below, see, e.g. \cite[\S VII.1.1-2]{EE}. Note that 
$\dom(-\Delta_{\rm free})=H^2(\dR^d)$ but the functions in $\dom(-\Delta_{\rm N})$ have only local $H^2$-regularity, that is,
$\dom(-\Delta_{\rm N})\subset H^2_{\rm loc}(\dR^d\setminus\Sigma)$.

\subsection{Definition of Schr\"{o}dinger operators with $\delta$ and $\delta^\prime$-interactions via sesquilinear forms}
\label{ssec:delta}

In this subsection we define Schr\"odinger operators with $\delta$ and $\delta^\prime$-interactions supported
on possibly non-compact boundaries of Lipschitz partitions with the help of corresponding
sesquilinear forms; cf. \cite{BEKS94} for the case of $\delta$-interactions and \cite{BLL13} for the case 
of $\delta'$-interactions on smooth hypersurfaces.
The domains of these operators are characterized and, in particular, 
the boundary conditions are given explicitly. For the special case of smooth domains with compact boundaries
the present description reduces to the one in  \cite{BLL13},  where 
a different approach via extension theory of symmetric operators and boundary triple techniques from \cite{BL07} was used. 
We also refer to \cite{AKMN13, AGS87, S88} for an approach via separation of variables in the case of 
spherically symmetric supports of interactions.

Let $\cP = \{\Omega_k\}_{k=1}^n$ be a Lipschitz partition of $\dR^d$ 
with the boundary $\Sigma$, let $\alpha,\beta:\Sigma\rightarrow\dR$ be such that $\alpha,\beta^{-1}\in L^\infty(\Sigma)$ 
and define the symmetric sesquilinear forms $\fra_{\delta,\alpha}$ and $\fra_{\delta',\beta}$
by 
\begin{equation}\label{deltaform}
\fra_{\delta,\alpha}[f, g] := \big(\nabla f, \nabla g\big)_{L^2(\dR^d;\dC^d)} - 
\big(\alpha f|_\Sigma, g|_\Sigma\big)_{L^2(\Sigma)},\quad \dom\fra_{\delta,\alpha} = H^1(\dR^d),
\end{equation}
and
\begin{equation}\label{delta'form}
\begin{split}
\fra_{\delta',\beta}[f,g] &:= \sum_{k=1}^n \big(\nabla f_k,\nabla g_k\big)_{L^2(\Omega_k;\dC^d)}\\  
&\quad-
\sum_{k=1}^{n-1}\sum_{l =k+1}^n 
\big(\beta^{-1}_{kl}(f_k|_{\Sigma_{kl}} - f_l|_{\Sigma_{kl}}), g_k|_{\Sigma_{kl}} - g_l|_{\Sigma_{kl}}\big)_{L^2(\Sigma_{kl})},\\
\dom \fra_{\delta',\beta} &= \bigoplus_{k=1}^n H^1(\Omega_k),
\end{split}
\end{equation}
respectively; here $\Sigma_{kl}=\partial\Omega_k\cap\partial\Omega_l$ for $k,l=1,2,\dots,n$, $k\ne l$, and $\beta_{kl}$ 
denotes the restrictions of $\beta$ to $\Sigma_{kl}$. The traces $f_k|_{\Sigma_{kl}}$ are
understood as restrictions of the trace $f_k|_{\partial\Omega_k}$ onto $\Sigma_{kl}$.
Note that $\sigma_k(\Sigma_{kl}) = \sigma_l(\Sigma_{kl}) = 0$ if the domains $\Omega_k$ and $\Omega_l$ are not neighbouring and that 
$$L^2(\Sigma)=\bigoplus_{k=1}^{n-1}\bigoplus_{l=k+1}^{n} L^2(\Sigma_{kl})\quad\text{and}\quad 
L^2(\partial\Omega_k) = \bigoplus_{l=1,\, l\ne k}^n L^2(\Sigma_{kl}).$$

\begin{prop}
\label{prop:deltaform}
The symmetric sesquilinear forms $\fra_{\delta,\alpha}$ and $\fra_{\delta',\beta}$ are
closed and semibounded from below. 
\end{prop}

\begin{proof}
We verify the assertion for $\fra_{\delta,\alpha}$ first. For this
note that $\fra_{\delta,\alpha}=\fra_{\rm free} + \fra^\prime$, where $\fra_{\rm free}$ is as in \eqref{freeNform} and 
\[
\fra'[f,g] := -(\alpha f,g)_{L^2(\Sigma)},\qquad \dom\fra' = H^1(\dR^d).
\] 
We show that $\fra'$ is bounded with 
respect to $\fra_{\rm free}$ with form bound $< 1$. In fact, for $f = \oplus_{k=1}^n f_k\in H^1(\dR^d)$ we have
\begin{equation}
\label{frq'}
\big|\fra'[f]\big| \le\|\alpha\|_\infty\big\|f|_\Sigma \big\|_{L^2(\Sigma)}^2 = \|\alpha\|_\infty\,\frac{1}{2}\sum_{k=1}^n \big\| f_k|_{\partial \Omega_k}\big\|^2_{L^2(\partial\Omega_k)}.
\end{equation}
According to Lemma~\ref{lem:trace} for any $\varepsilon > 0$ and  $k=1,2,\dots,n$,
there exists $C_k(\varepsilon) > 0$ such that 
\begin{equation}\label{inequality_delta1}
\big\| f_k|_{\partial\Omega_k}\big\|_{L^2(\partial\Omega_k)}^2 \le \varepsilon \|\nabla f_k\|^2_{L^2(\Omega_k;\dC^d)} 
+ C_k(\varepsilon)\|f_k\|_{L^2(\Omega_k)}^2.
\end{equation}
Therefore \eqref{frq'} yields
\begin{equation*}
 \begin{split}
  \big|\fra'[f]\big|& \le \Vert\alpha\Vert_\infty\,\frac{\varepsilon}{2}\sum_{k=1}^n\|\nabla f_k\|^2_{L^2(\Omega_k;\dC^d)}  + 
  \Vert\alpha\Vert_\infty\,\frac{1}{2}\sum_{k=1}^n C_k(\varepsilon)\|f_k\|_{L^2(\Omega_k)}^2\\
  & \le \|\alpha\|_\infty\,\frac{\varepsilon}{2}\,\fra_{\rm free}[f] + 
\|\alpha\|_\infty\,\frac{\max_k C_k(\varepsilon)}{2}\,\|f\|_{L^2(\dR^d)}^2
 \end{split}
\end{equation*}
for all $f\in \dom \fra' = \dom \fra_{\rm free}$. 
Thus, for sufficiently small $\varepsilon$ the form $\fra'$ is form bounded with respect to the form $\fra_{\rm free}$ with 
form bound $< 1$. Then by \cite[VI~Theorem~1.33]{Kato} the form 
$\fra_{\delta,\alpha}$ is closed and semibounded from below.  

Next we prove the statement for $\fra_{\delta',\beta}$. As above we have
$\fra_{\delta',\beta}=\fra_{\rm N} +\fra''$, where $\fra_{\rm N}$ is as in \eqref{freeNform} and 
\begin{equation*}
 \begin{split}
  \fra''[f,g] &:= -
\sum_{k=1}^{n-1}\sum_{l =k+1}^n 
\big(\beta^{-1}_{kl}(f_k|_{\Sigma_{kl}} - f_l|_{\Sigma_{kl}}), g_k|_{\Sigma_{kl}} - g_l|_{\Sigma_{kl}}\big)_{L^2(\Sigma_{kl})},\\
 \dom\fra'' &= \bigoplus_{k=1}^n H^1(\Omega_k).
 \end{split}
\end{equation*}
We show that $\fra''$ is bounded with 
respect to $\fra_{\rm N}$ with form bound $< 1$. In fact, 
\begin{equation*}
 \begin{split}
  \big|\fra''[f]\big| &\le \|\beta^{-1}\|_\infty \sum_{k=1}^{n-1}\sum_{l =k+1}^n \big\|f_k|_{\Sigma_{kl}} - 
  f_l|_{\Sigma_{kl}}\big\|_{L^2(\Sigma_{kl})}^2\\
 &\le 2\|\beta^{-1}\|_\infty \sum_{k=1}^{n-1}\sum_{l =k+1}^n\Bigl(\,\big\|f_k|_{\Sigma_{kl}}\big\|_{L^2(\Sigma_{kl})}^2 + 
  \big\|f_l|_{\Sigma_{kl}}\big\|_{L^2(\Sigma_{kl})}^2\,\Bigr)\\
  &=2 \|\beta^{-1}\|_\infty\sum_{k=1}^n \big\| f_k|_{\partial\Omega_k}\big\|_{L^2(\partial\Omega_k)}^2
 \end{split}
\end{equation*}
and with the help of \eqref{inequality_delta1} (see Lemma~\ref{lem:trace}) we conclude that for any $\varepsilon > 0$ and 
$k=1,2,\dots,n$, 
there exists $C_k(\varepsilon) > 0$ such that
\begin{equation*}
 \begin{split}
  |\fra''[f]| &\le 2\varepsilon\,\|\beta^{-1}\|_\infty \sum_{k=1}^n \|\nabla f_k\|^2_{L^2(\Omega_k;\dC^d)} +  
   2\|\beta^{-1}\|_\infty \sum_{k=1}^n C_k(\varepsilon)\|f_k\|_{L^2(\Omega_k)}^2\\
   &\le 2\varepsilon\,\|\beta^{-1}\|_\infty\, \fra_{\rm N}[f] + 2\|\beta^{-1}\|_\infty\,\max_k C_k(\varepsilon)\,\|f\|_{L^2(\dR^d)}^2
 \end{split}
\end{equation*}
for all $f\in \dom\fra'' = \dom\fra_{\rm N}$. Hence for  $\varepsilon >0$ sufficiently small $\fra''$ is bounded with 
respect to $\fra_{\rm N}$
with form bound $< 1$. As above it follows from 
\cite[VI Theorem 1.33]{Kato} that $\fra_{\delta^\prime,\beta}$ is closed and semibounded from below.  
\end{proof}

It follows from Proposition~\ref{prop:deltaform} and the first representation theorem \cite[VI Theorem 2.1]{Kato} that 
there are unique self-adjoint operators $-\Delta_{\delta,\alpha}$ and $-\Delta_{\delta^\prime,\beta}$  in $L^2(\dR^d)$ 
associated with the sesquilinear forms $\fra_{\delta,\alpha}$ and $\fra_{\delta^\prime,\beta}$, respectively,
such that
$$(-\Delta_{\delta,\alpha} f,g)_{L^2(\dR^d)}=\fra_{\delta,\alpha}[f,g]\quad\text{and}\quad
(-\Delta_{\delta^\prime,\beta} f,g)_{L^2(\dR^d)}=\fra_{\delta^\prime,\beta}[f,g]
$$
for $f\in\dom(-\Delta_{\delta,\alpha})\subset\dom \fra_{\delta,\alpha}$, $g\in\dom\fra_{\delta,\alpha}$,
and $f\in\dom(-\Delta_{\delta^\prime,\beta})\subset\dom \fra_{\delta^\prime,\beta}$, $g\in\dom\fra_{\delta^\prime,\beta}$,
respectively.

\begin{dfn}\label{def:delta'}
The selfadjoint operator  $-\Delta_{\delta,\alpha}$ {\rm (}$-\Delta_{\delta',\beta}${\rm )} in $L^2(\dR^d)$  
is called Schr\"o\-dinger operator with $\delta$-interaction of strength $\alpha$ {\rm (}$\delta'$-interaction 
of strength $\beta$, respectively{\rm )} 
supported on $\Sigma$. 
\end{dfn}

Observe that by the definition of the forms $\fra_{\delta,\alpha}$ and $\fra_{\delta^\prime,\beta}$ the $\delta$-interaction
is strong if $\alpha$ is big, and the
$\delta^\prime$-interaction is strong if $\beta$ is small.

In the next theorem we characterize the action and domain of $-\Delta_{\delta,\alpha}$ and $-\Delta_{\delta',\beta}$.

\begin{thm}
\label{thm:delta}
Let $\cP = \{\Omega_k\}_{k=1}^n$ be a Lipschitz partition of $\dR^d$ 
with the boundary $\Sigma$, let $\alpha,\beta:\Sigma\rightarrow\dR$ be such that $\alpha,\beta^{-1}\in L^\infty(\Sigma)$ 
and
let $-\Delta_{\delta,\alpha}$ and $-\Delta_{\delta',\beta}$ be the self-adjoint operators associated with $\fra_{\delta,\alpha}$
and $\fra_{\delta',\beta}$, respectively.
Then the following holds.
\begin{itemize}
 \item [{\rm (i)}] $-\Delta_{\delta,\alpha}f = \oplus_{k=1}^n (-\Delta f_k)$ and $f= \oplus_{k=1}^n f_k\in\dom(-\Delta_{\delta,\alpha})$  if and only if
\vskip 0.2cm
 \begin{itemize}\setlength{\itemsep}{0.2cm}
\item[\rm (a)] $f\in H^1(\dR^d)$,
\item[\rm (b)] $\Delta f_k  \in L^2(\Omega_k)$ for all $k =1,2,\dots, n$, 
\item[\rm (c)] $\partial_\cP f|_\Sigma\in L^2(\Sigma)$ exists in the sense of Definition~\ref{dfn:deltatrace} 
and 
\[
\partial_\cP f|_\Sigma = \alpha f|_{\Sigma}.
\]
\end{itemize}
 \item [{\rm (ii)}] $-\Delta_{\delta^\prime,\beta}f = \oplus_{k=1}^n (-\Delta f_k)$ and $f= \oplus_{k=1}^n f_k\in\dom(-\Delta_{\delta^\prime,\beta})$ 
   if and only if
\vskip 0.2cm
\begin{itemize}\setlength{\itemsep}{0.2cm}
\item[\rm (a$^\prime$)] $f_k \in H^1(\Omega_k)$ for all $k=1,2,\dots,n$,
\item[\rm (b$^\prime$)] $\Delta f_k  \in L^2(\Omega_k)$ for all $k=1,2,\dots,n$,
\item[\rm (c$^\prime$)] $\partial_{\nu_k}f_k|_{\partial\Omega_k}\in L^2(\partial\Omega_k)$ for all $k=1,2,\dots,n$ in the sense of 
Definition~\ref{dfn:LaplacianNeumann}~{\rm (ii)} and 
\begin{equation*}
f_k|_{\partial\Omega_k} - \oplus_{ l\ne k} f_l|_{\Sigma_{kl}} 
= \beta_k\partial_{\nu_k}f_k|_{\partial\Omega_k}, \quad k=1,2,\dots,n.
\end{equation*}
\end{itemize}
\end{itemize}
\end{thm}

\begin{proof}
The proof of items (i) and (ii) consists of three steps each. First we show that $-\Delta_{\delta,\alpha}$ 
and $-\Delta_{\delta^\prime,\beta}$
act as minus Laplacians on each $\Omega_k$. In the second step we verify that
$f\in\dom(-\Delta_{\delta,\alpha})$ ($f\in\dom(-\Delta_{\delta^\prime,\beta})$) satisfies the conditions (a)--(c) ((a$^\prime$)--(c$^\prime$), 
respectively), and in the last step we prove the converse implication.
\vskip 0.2cm 
\noindent 
(i) {\bf Step I.}  Let $f\in\dom(-\Delta_{\delta,\alpha})$, $g_k\in H^1_0(\Omega_k)$ for some $k=1,2,\dots,n$, and
extend $g_k$ by zero to $\wt g_k\in H^1(\dR^d)=\dom\fra_{\delta,\alpha}$. From $\wt g_k|_\Sigma = 0$ and the first representation theorem  we obtain
\[
\begin{split}
\bigl((-\Delta_{\delta,\alpha}f)_k,g_k\bigr)_{L^2(\Omega_k)} &= (-\Delta_{\delta,\alpha}f,\wt g_k)_{L^2(\dR^d)} = \fra_{\delta,\alpha}[f,\wt g_k] \\
&= (\nabla f, \nabla \wt g_k)_{L^2(\dR^d;\dC^d)} - (\alpha f|_\Sigma,\wt g_k|_\Sigma)_{L^2(\Sigma)} \\
&=(\nabla f, \nabla \wt g_k)_{L^2(\dR^d;\dC^d)} = (\nabla f_k, \nabla  g_k)_{L^2(\Omega_k;\dC^d)}.
\end{split}
\]
Therefore, by Definition~\ref{dfn:LaplacianNeumann}~(i) we have
$
(-\Delta_{\delta,\alpha}f)_k = -\Delta f_k \in L^2(\Omega_k)$ for all $k=1,2,\dots,n$, that is, 
\[
-\Delta_{\delta,\alpha}f = \oplus_{k=1}^n (-\Delta f_k) \in L^2(\dR^d).
\]

\noindent 
{\bf Step II.} Let $f$ be a function in $\dom(-\Delta_{\delta,\alpha})$.
Then $f$ satisfies condition (a) since 
$\dom(-\Delta_{\delta,\alpha})\subset \dom\fra_{\delta,\alpha} = H^1(\dR^d)$.
Condition (b) is satisfied as we have shown in Step I. Hence it remains to
check condition (c). For this let $h \in \dom\fra_{\delta,\alpha}$.
From Step I and the first representation theorem we conclude 
\begin{equation*}
\begin{split}
\bigl(\oplus_{k=1}^n(-\Delta f_k), h\bigr)_{L^2(\dR^d)}&=(-\Delta_{\delta,\alpha}f,h)_{L^2(\dR^d)}
= \fra_{\delta,\alpha}[f,h] \\
&= (\nabla f,\nabla h)_{L^2(\dR^d;\dC^d)}
- (\alpha f|_{\Sigma}, h|_\Sigma)_{L^2(\Sigma)}
\end{split}
\end{equation*}
which yields 
\begin{equation*}
(\nabla f,\nabla h)_{L^2(\dR^d;\dC^d)} - \bigl(\oplus_{k=1}^n(-\Delta f_k), h\bigr)_{L^2(\dR^d)} =
(\alpha f|_{\Sigma}, h|_\Sigma)_{L^2(\Sigma)}.
\end{equation*}
Hence, by Definition~\ref{dfn:deltatrace} we have $\partial_\cP f|_\Sigma\in L^2(\Sigma)$ 
and 
$\partial_\cP f|_\Sigma= \alpha f|_\Sigma$,
that is, $f$ satisfies condition (c).

\vskip 0.2cm
\noindent {\bf Step III.}
Assume that $f$ satisfies the conditions (a)-(c)
and let $h\in\dom\fra_{\delta,\alpha}$.
By condition (a) we have
$f\in\dom\fra_{\delta,\alpha}$ and hence
\[
\fra_{\delta,\alpha}[f,h] = (\nabla f,\nabla h)_{L^2(\dR^d;\dC^d)}
- (\alpha f|_{\Sigma}, h|_\Sigma)_{L^2(\Sigma)}.
\]
The conditions  (b) and (c) together with Definition~\ref{dfn:deltatrace} imply that 
$$
(\nabla f,\nabla h)_{L^2(\dR^d;\dC^d)}- \bigl(\oplus_{k=1}^n (-\Delta f_k), h\bigr)_{L^2(\dR^d)}= 
(\partial_\cP f|_{\Sigma}, h|_\Sigma)_{L^2(\Sigma)}=
(\alpha f|_{\Sigma}, h|_\Sigma)_{L^2(\Sigma)}
$$
and hence
$
\fra_{\delta,\alpha}[f,h] = (\oplus_{k=1}^n (-\Delta f_k), h)_{L^2(\dR^d)}
$
for all $h\in\dom\fra_{\delta,\alpha}$. The first representation theorem
yields $f\in \dom(-\Delta_{\delta,\alpha})$.

\vskip 0.2cm\noindent
(ii) {\bf Step I.} The same reasoning as in (i) Step I yields
$-\Delta_{\delta',\beta}f  =\oplus_{k=1}^n(-\Delta f_k)\in L^2(\dR^d)$ for all $f\in  \dom(-\Delta_{\delta',\beta})$.

\vskip 0.2cm
\noindent {\bf Step II.} Let $f$ be a function in $\dom(-\Delta_{\delta^\prime,\beta})$.
Then $f$ satisfies condition (a$^\prime$) since 
$\dom(-\Delta_{\delta^\prime,\beta})\subset \dom\fra_{\delta^\prime,\beta}$
and condition (b$^\prime$) holds by Step I. We
check condition (c$^\prime$). For this let $h_k\in H^1(\Omega_k)$ for
some $k= 1,2,\dots,n$, and let $\wt h_k\in\dom\fra_{\delta',\beta}$ 
be its extension by zero.
From Step I and the first representation theorem we conclude 
\begin{equation} \label{fra4}
\begin{split}
&(-\Delta f_k, h_k)_{L^2(\Omega_k)}=(-\Delta_{\delta',\beta}f,\wt h_k)_{L^2(\dR^d)} =\fra_{\delta',\beta}[f,\wt h_k]\\
&\quad =(\nabla f_k, \nabla h_k)_{L^2(\Omega_k;\dC^d)}- 
\sum_{l=1,\, l\ne k}^n \bigl(\beta_{kl}^{-1}(f_k|_{\Sigma_{kl}} -f_l|_{\Sigma_{kl}}), h_k|_{\Sigma_{kl}}\bigr)_{L^2(\Sigma_{kl})}
\end{split}
\end{equation}
where we used that $\wt h_k|_{\Sigma_{pq}} =0$ if $k\not= p,q$.
For $k=1,\dots,n$ we set
\[
b_k := \bigoplus_{l=1,\,l\ne k}^n\big(\beta_{kl}^{-1}(f_k|_{\Sigma_{kl}} -f_l|_{\Sigma_{kl}})\big)\in \bigoplus_{l=1,\,l\ne k}^n
L^2(\Sigma_{kl})\,=\,L^2(\partial\Omega_k).
\]
From \eqref{fra4} we then obtain
\begin{equation*}
(\nabla f_k,  \nabla h_k)_{L^2(\Omega_k;\dC^d)} - (-\Delta f_k,h_k)_{L^2(\Omega_k)} = (b_k, h_k|_{\partial\Omega_k})_{L^2(\partial\Omega_k)}
\end{equation*}
for all $h_k\in H^1(\Omega_k)$ and $k=1,\dots,n$.
Hence, $\partial_{\nu_k}f_k|_{\partial\Omega_k}\in L^2(\partial\Omega_k)$ exists in the sense of Definition~\ref{dfn:LaplacianNeumann}~(ii) and
the boundary condition
\[
\beta_k \partial_{\nu_k}f_k|_{\partial\Omega_k} = \beta_k b_k = f_k|_{\partial\Omega_k} - \bigoplus_{l=1,\, l\ne k}^n f_l|_{\Sigma_{kl}},
\qquad k=1,\dots, n,
\]
holds, that is, condition (c$^\prime$) is valid for all $f\in \dom(-\Delta_{\delta^\prime,\beta})$. \\[0.2ex]

\vskip 0.2cm
\noindent {\bf Step III.} 
Assume that $f$ satisfies conditions (a$^\prime$)-(c$^\prime$), and let $h\in\dom\fra_{\delta',\beta}$. 
Fix some $k=1,\dots,n$ and let $\widetilde h_k$ be the extension of $h_k\in H^1(\Omega_k)$ by zero.  
By condition (a$^\prime$) $f\in\dom\fra_{\delta',\beta}$ and hence 
$$
\fra_{\delta',\beta}[f,\widetilde h_k] = (\nabla f_k,\nabla h_k)_{L^2(\Omega_k;\dC^d)} -\sum_{l=1,\,l\ne k}^n \bigl(\beta_{kl}^{-1}\,(f_k|_{\Sigma_{kl}}-f_l|_{\Sigma_{kl}}), h_k|_{\Sigma_{kl}}\bigr)_{L^2(\Sigma_{kl})}.
$$
On the other hand Definition~\ref{dfn:LaplacianNeumann}~(ii) and conditions (b$^\prime$) and (c$^\prime$) imply
\begin{equation*}
\begin{split}
&(\nabla f_k,\nabla h_k)_{L^2(\Omega_k;\dC^d)} - (-\Delta f_k, h_k)_{L^2(\Omega_k)}\\
&\qquad\qquad\qquad=\bigl(\beta_k^{-1}\, (f_k|_{\partial\Omega_k} - 
\oplus_{ l\ne k} f_l|_{\Sigma_{kl}}), h_k|_{\partial\Omega_k}\bigr)_{L^2(\partial\Omega_k)}
\end{split}
\end{equation*}
and hence $\fra_{\delta',\beta}[f,\widetilde h_k] = (-\Delta f_k,h_k)_{L^2(\Omega_k)} $ for 
$k=1,2,\dots, n$. Summing up we conclude
\begin{equation*}
  \fra_{\delta',\beta}[f,h]  = \sum_{k=1}^n \fra_{\delta',\beta}[f,\widetilde h_k] =
\sum_{k=1}^n (-\Delta f_k,h_k)_{L^2(\Omega_k)}
 =\bigl(\oplus_{k=1}^n(-\Delta f_k),h\bigr)_{L^2(\dR^d)}
\end{equation*}
for any $h\in\dom\fra_{\delta',\beta}$. This implies $f\in\dom(-\Delta_{\delta',\beta})$.
\end{proof}

We remark that the condition 
$
\partial_\cP f|_\Sigma = \alpha f|_{\Sigma}
$
for the functions in $\dom(-\Delta_{\delta,\alpha})$ in Theorem~\ref{thm:delta}~(i)-(c) reflects the classical $\delta$-jump boundary 
conditions on common boundaries of the domains in 
the partition; cf. \cite[I. Theorem 3.1.1]{AGHH} and Remark~\ref{remdeltalocal}. Similarly the condition 
\[
f_k|_{\partial\Omega_k} - \oplus_{l\ne k}f_l|_{\Sigma_{kl}} = \beta_k\partial_{\nu_k}f_k|_{\partial\Omega_k},\qquad k=1,2,\dots,n,
\]
for the functions in $\dom(-\Delta_{\delta^\prime,\beta})$ in Theorem~\ref{thm:delta}~(ii)-(c$^\prime$)
corresponds to the classical $\delta'$-jump boundary conditions; cf. \cite[I. equation (4.5)]{AGHH}. 
Note also that our sign choice for $\alpha$ and $\beta$ in the definition of the forms in \eqref{deltaform}-\eqref{delta'form} 
and the associated operators is opposite with respect to \cite{AGHH}.

Observe that for a function $f$ in $\dom(-\Delta_{\delta,\alpha})$ or $\dom(-\Delta_{\delta^\prime,\beta})$ it follows from 
Theorem~\ref{thm:delta}~(i)-(b), (ii)-(b$^\prime$) and elliptic regularity that 
$$f_k\in H^2_{\rm loc}(\Omega_k),\qquad k=1,\dots,n.$$
It is not surprising that additional assumptions on the smoothness of the boundary (or parts of the boundary) 
and the coefficients $\alpha,\beta^{-1}$ lead to $H^2$-regularity of $f$ up to the boundary (or parts of it, respectively). 
We first recall a result from \cite{BLL13} for a particular smooth partition and turn to a more general situation
in the lemma below.

\begin{prop}
Let $\Omega$ be a bounded domain with $C^\infty$-boundary and consider the partition 
$\cP = \{\Omega, \dR^d\setminus \ov{\Omega}\}$ with boundary $\Sigma=\partial\Omega$.
Then the following holds.
\begin{itemize}
 \item [{\rm (i)}] If $\alpha\in L^\infty(\Sigma)$ then both domains
                   $\dom(-\Delta_{\delta,\alpha})$ and
                   $\dom(-\Delta_{\delta',\beta})$ are contained in $H^{3/2}(\Omega)\oplus H^{3/2}(\dR^d\setminus\ov\Omega)$.
 \item [{\rm (ii)}] If $\alpha\in W^{1,\infty}(\Sigma)$ then both domains
                   $\dom(-\Delta_{\delta,\alpha})$ and
                   $\dom(-\Delta_{\delta',\beta})$ are contained in $H^2(\Omega)\oplus H^2(\dR^d\setminus\ov\Omega)$.
 \end{itemize}
\end{prop}

In the next lemma we establish local $H^2$-regularity up to parts of the boundary $\Sigma$ of a Lipschitz partition
$\cP = \{\Omega_k\}_{k=1}^n$ under the assumption that the corresponding part of the boundary and $\alpha,\beta^{-1}\in L^\infty(\Sigma)$
are locally $C^{1,1}$ and $C^1$, respectively. This observation, which is essentially a consequence of the boundary conditions
in Theorem~\ref{thm:delta}~(i)-(c), (ii)-(c$^\prime$) and \cite[Theorem 4.18, Theorem 4.20]{McLean}, will be used in the proof of 
Theorem~\ref{thm:compact}.

\begin{lem}\label{h2bound}
Let $\Omega_k$ and $\Omega_l$ be neighbouring domains of a Lipschitz partition $\cP = \{\Omega_k\}_{k=1}^n$ and
let $\Gamma$ be a bounded open subset of $\Sigma_{kl}=\partial\Omega_k\cap\partial\Omega_l$ which is $C^{1,1}$ and part of a Lipschitz dissection of $\partial\Omega_k$
 in the sense of \cite[page 99]{McLean}. Assume that $\Gamma\cap\Sigma_{km}=\varnothing$
 for $m=1,\dots,n$ with $m\not=l,k$. Then for any relatively open subset $\gamma$ of $\Gamma$ with $\overline\gamma\subset\Gamma$ there exists an open set 
 $G_{kl}$ of $\Omega_k\cup\Omega_l\cup\Gamma$ such that $\overline\gamma\subset G_{kl}$ and the following holds for $j=k,l$.
 \begin{itemize}
  \item [{\rm (i)}] If $\alpha|_{\Gamma}\in C^1(\Gamma)$ and $f\in\dom(-\Delta_{\delta,\alpha})$ then $f_j\in H^2(\Omega_j\cap G_{kl})$;
  \item [{\rm (ii)}] If $\beta^{-1}|_{\Gamma}\in C^1(\Gamma)$ and $f\in\dom(-\Delta_{\delta^\prime,\beta})$ then 
   $f_j\in H^2(\Omega_j\cap G_{kl})$. \end{itemize}
\end{lem}

\begin{proof}
(i) For $f\in\dom(-\Delta_{\delta,\alpha})$ we have $f_k\in H^1(\Omega_k)$ and hence $f_k\vert_{\Gamma}\in H^{1/2}(\Gamma)$.
 Therefore $\alpha|_{\Gamma}\in C^1(\Gamma)$ yields $\alpha f_k\vert_{\Gamma}\in H^{1/2}(\Gamma)$ by \cite[Theorem 3.20]{McLean}.
 The boundary condition in Theorem~\ref{thm:delta}~(i)-(c) and its local form in \eqref{deltalocal} 
 (interpreted in $H^{-1/2}(\Gamma)$ if the Neumann traces $\partial_{\nu_k}u_k\vert_\Gamma$ and 
 $\partial_{\nu_l}u_l\vert_\Gamma$ do not exist in $L^2(\Gamma)$) 
 together with \cite[Theorem 4.20]{McLean} implies the statement.
  
(ii) As in the proof of (i) we have $f_k\vert_{\Gamma}\in H^{1/2}(\Gamma)$ for $f\in\dom(-\Delta_{\delta^\prime,\beta})$ and the assumption
$\beta^{-1}|_{\Gamma}\in C^1(\Gamma)$ together with Theorem~\ref{thm:delta}~(ii)-(c$^\prime$) and 
\cite[Theorem~3.20]{McLean} implies
$\beta^{-1}\bigl( f_k|_\Gamma - f_l|_\Gamma\bigr) 
=  \partial_{\nu_k}f_k|_{\Gamma} \in H^{1/2}(\Gamma).$
Now the assertion follows from \cite[Theorem 4.18~(ii)]{McLean}.
\end{proof}

\subsection{An operator inequality for Schr\"{o}dinger operators with $\delta$ and $\delta^\prime$-interactions}
\label{ssec:ineq}
Let again $\cP = \{\Omega_k\}_{k=1}^n$ be a Lipschitz partition of $\dR^d$ with boundary $\Sigma$, 
and let $\alpha,\beta:\Sigma\rightarrow\dR$ be such that $\alpha,\beta^{-1}\in L^\infty(\Sigma)$.
In Theorem~\ref{thm:main} below we prove an operator inequality for the Schr\"{o}dinger operators $-\Delta_{\delta,\alpha}$
and $-\Delta_{\delta',\beta}$ which is intimately 
related with the chromatic number $\chi$ of the partition $\cP$. 
\begin{thm}
\label{thm:main}
Let $\cP = \{\Omega_k\}_{k=1}^n$ be a Lipschitz partition of $\dR^d$ 
with boundary $\Sigma$ and chromatic number $\chi$. Let $\alpha,\beta\colon\Sigma\rightarrow\dR$
be such that $\alpha,\beta^{-1}\in L^\infty(\Sigma)$ and assume 
that
\begin{equation}
\label{mainassump}
0 < \beta \le \frac{4}{\alpha}\sin^2\big(\pi / \chi\big).
\end{equation}
Then there exists a unitary operator $U\colon L^2(\dR^d)\rightarrow L^2(\dR^d)$
such that the self-adjoint operators $-\Delta_{\delta,\alpha}$ and $-\Delta_{\delta',\beta}$ satisfy the inequality
\[
U^{-1}(-\Delta_{\delta',\beta})U \le -\Delta_{\delta,\alpha}.
\]
\end{thm}

\begin{proof}
By the definition of the chromatic number (Definition~\ref{dfn:chro}) 
there exists an optimal colouring mapping
\[
\phi\colon \{1,2,\dots,n\}\rightarrow \{0,1,\dots,\chi-1\}
\]
such that for any $k,l = 1,2,\dots,n$, $k\ne l$, we have
\[
\sigma_{k}(\Sigma_{kl}) > 0 \quad \Longrightarrow\quad \phi(k) \ne \phi(l).
\]
Next, we define $n$ complex numbers $\cZ := \{z_k\}_{k=1}^n$ on the unit circle by
\[
z_k := \exp\Big(\tfrac{2\pi\phi(k)}{\chi}i\Big), \qquad k=1,2,\dots,n.
\]
Among the $z_k$ there are only $\chi$ distinct numbers.
The points $z_k$, $k=1,\dots,n$, on the unit circle form the vertices 
of an equilateral polygon with $\chi$ edges. The square of the length of these edges is
\begin{equation}\label{edgelength}
2 - 2\cos\big(2\pi /\chi \big)= 4\sin^2\big(\pi /\chi\big).
\end{equation}
Observe that for any $k,l=1,2,\dots,n$, $k\ne l$, with 
$\sigma_k(\Sigma_{kl}) >0$ we have 
\begin{equation*}
\begin{split}
|z_k -z_l|^2 &= \Big[\cos\Big(\tfrac{2\pi\phi(k)}{\chi}\Big) - \cos\Big(\tfrac{2\pi\phi(l)}{\chi}\Big)\Big]^2 + \Big[\sin\Big(\tfrac{2\pi\phi(k)}{\chi}\Big) - \sin\Big(\tfrac{2\pi\phi(l)}{\chi}\Big)\Big]^2 \\
&= 2-2 \cos\Big(\tfrac{2\pi\phi(k)}{\chi}\Big)\cos\Big(\tfrac{2\pi\phi(l)}{\chi}\Big)-2 \sin\Big(\tfrac{2\pi\phi(k)}{\chi}\Big)\sin\Big(\tfrac{2\pi\phi(l)}{\chi}\Big)\\
&= 2 - 2\cos\Big(\tfrac{2\pi}{\chi}(\phi(k) - \phi(l))\Big) \ge 2 - 2\cos\big(2\pi /\chi \big),
\end{split}
\end{equation*}                                                  
where we used standard trigonometric identities in the third equality and $$\phi(k)-\phi(l)\in\{-(\chi-1),\dots,\chi-1\}\setminus\{0\}$$ 
in the last estimate.
Together with \eqref{edgelength} we find 
$4\sin^2 (\pi / \chi )\leq |z_k -z_l|^2$ and hence by the assumption \eqref{mainassump}
\begin{equation}\label{home}
 0 <\alpha\leq \frac{4}{\beta} \sin^2\big(\pi / \chi \big) \leq \alpha_{\cZ},
\end{equation}
where $\alpha_{\cZ}(x)= |z_k - z_l|^2\beta_{kl}^{-1}(x)$, $x\in\Sigma_{kl}$, and $k,l =1,2,\dots,n$, $ k\ne l$.

Define a unitary mapping $U_{\cZ}\colon L^2(\dR^d) \rightarrow L^2(\dR^d)$ by 
$$(U_{\cZ}f)(x) :=  z_k f_k(x),\qquad x\in\Omega_k,\qquad k=1,\dots,n,$$
and a corresponding sesquilinear form $\wt\fra_{\delta',\beta}$ by
\[
\wt\fra_{\delta',\beta}[f,g] := \fra_{\delta',\beta}[U_\cZ f, U_\cZ g],\qquad \dom\wt\fra_{\delta',\beta} =\dom\fra_{\delta',\beta}.
\]
Observe that $\wt\fra_{\delta',\beta}$ is a closed, densely defined, symmetric form which is semibounded from below, and that the selfadjoint operator
associated with $\wt\fra_{\delta',\beta}$ is 
\begin{equation*}
U_\cZ^{-1}(-\Delta_{\delta',\beta}) U_\cZ, \qquad \dom\bigl(U_\cZ^{-1}(-\Delta_{\delta',\beta}) U_\cZ\bigr)=U_\cZ^{-1}\bigl(\dom(-\Delta_{\delta',\beta})\bigr).
\end{equation*}

We claim that the inequality $\wt\fra_{\delta',\beta}\leq \fra_{\delta,\alpha_\cZ}$ holds. In fact, 
$$\dom\fra_{\delta,\alpha_{\cZ}}=H^1(\dR^d)\subset\bigoplus_{k=1}^n H^1(\Omega_k)= \dom\wt\fra_{\delta',\beta}$$
is clear and for $f\in\dom\fra_{\delta,\alpha_{\cZ}}$ we have $f\vert_{\Sigma_{kl}}=f_k\vert_{\Sigma_{kl}}=f_l\vert_{\Sigma_{kl}}$. Therefore
we obtain
\begin{equation*}
 \begin{split}
  \wt\fra_{\delta',\beta}[f] & =\fra_{\delta',\beta}[U_\cZ f ]\\ 
  & = \sum_{k=1}^n \Vert z_k\nabla  f_k\Vert^2_{L^2(\Omega_k;\dC^d)}-
\sum_{k=1}^{n-1}\sum_{l =k+1}^n 
\vert z_k-z_l\vert^2 \big(\beta^{-1}_{kl} f|_{\Sigma_{kl}}, f|_{\Sigma_{kl}}\big)_{L^2(\Sigma_{kl})}\\
& = \|\nabla f\|^2_{L^2(\dR^d;\dC^d)} - (\alpha_\cZ f\vert_\Sigma,f\vert_\Sigma)_{L^2(\Sigma)}=\fra_{\delta,\alpha_{\cZ}}[f]
\end{split}
\end{equation*}
for all $f\in\dom\fra_{\delta,\alpha_{\cZ}}$, and hence $\wt\fra_{\delta',\beta}\leq \fra_{\delta,\alpha_\cZ}$. 
Moreover, as $\alpha\leq\alpha_\cZ$ by \eqref{home} we also have $\fra_{\delta,\alpha_\cZ}\leq \fra_{\delta,\alpha}$. This implies
$$\wt\fra_{\delta',\beta}\leq \fra_{\delta,\alpha}$$ and hence $U_\cZ^{-1}(-\Delta_{\delta',\beta})U_\cZ \le -\Delta_{\delta,\alpha}$.
\end{proof}

As an immediate consequence of Theorem~\ref{thm:main} and Theorem~\ref{thm:variation} we obtain the 
following corollary on the relation of the spectra of
$-\Delta_{\delta,\alpha}$ and $-\Delta_{\delta',\beta}$.

\begin{cor}\label{thm:main2}
Let $\cP = \{\Omega_k\}_{k=1}^n$ be a Lipschitz partition of $\dR^d$ 
with boundary $\Sigma$. Let $\alpha,\beta\colon\Sigma\rightarrow\dR$
be such that $\alpha,\beta^{-1}\in L^\infty(\Sigma)$ and assume 
that 
\begin{equation*}
0 < \beta \le \frac{4}{\alpha}\sin^2\big(\pi / \chi\big).
\end{equation*}
Denote by $\{\lambda_k(-\Delta_{\delta,\alpha})\}_{k=1}^\infty$ and $\{\lambda_k(-\Delta_{\delta',\beta})\}_{k=1}^\infty$
the eigenvalues of the operators $-\Delta_{\delta,\alpha}$ and $-\Delta_{\delta',\beta}$, respectively, below the bottom
of their essential spectra, enumerated in non-decreasing order and repeated with multiplicities, and let 
$N(-\Delta_{\delta,\alpha})$ and $N(-\Delta_{\delta',\beta})$ be their total numbers as in Definition~\ref{dfn:spec1}.
Then the following holds.
\begin{itemize}\setlength{\itemsep}{1.2ex}
\item [\rm (i)] $\lambda_k(-\Delta_{\delta',\beta})\le \lambda_k(-\Delta_{\delta,\alpha})$ 
for all $k\in\dN$;
\item [\rm (ii)] $\min\sess(-\Delta_{\delta',\beta})\le \min\sess(-\Delta_{\delta,\alpha})$;
\item [\rm (iii)] If $\min\sess(-\Delta_{\delta,\alpha})=\min\sess(-\Delta_{\delta',\beta})$
then $N(-\Delta_{\delta,\alpha})\leq N(-\Delta_{\delta',\beta})$.
\end{itemize}
\end{cor}

According to the four colour theorem the chromatic number of a Lipschitz partition $\cP$ of $\dR^2$ is $\chi\le 4$;
cf. \cite{AH77, AHK77} or \cite[\S 8.2]{MT01}. This implies the following corollary in the case $d=2$.

\begin{cor}
Let $\cP = \{\Omega_k\}_{k=1}^n$ be a Lipschitz partition of $\dR^2$ 
with boundary $\Sigma$ and chromatic number $\chi$. Let $\alpha,\beta\colon\Sigma\rightarrow\dR$
be such that $\alpha,\beta^{-1}\in L^\infty(\Sigma)$ and assume 
that
\begin{equation*}
0 < \beta \le \frac{2}{\alpha}.
\end{equation*}
Then there exists a unitary operator $U\colon L^2(\dR^2)\rightarrow L^2(\dR^2)$
such that the self-adjoint operators $-\Delta_{\delta,\alpha}$ and $-\Delta_{\delta',\beta}$ satisfy the inequality
\[
U^{-1}(-\Delta_{\delta',\beta})U \le -\Delta_{\delta,\alpha},
\]
and hence the assertions in Corollary~\ref{thm:main2} hold.
\end{cor}

For the case of a  Lipschitz partition with chromatic number $\chi=2$ Theorem~\ref{thm:main} reads as follows.

\begin{cor}
\label{cor2}
Let $\cP = \{\Omega_k\}_{k=1}^n$ be a Lipschitz partition of $\dR^d$ 
with boundary $\Sigma$ and chromatic number $\chi=2$. Let $\alpha,\beta\colon\Sigma\rightarrow\dR$
be such that $\alpha,\beta^{-1}\in L^\infty(\Sigma)$ and assume 
that
\begin{equation*}
0 < \beta \le \frac{4}{\alpha}.
\end{equation*}
Then there exists a unitary operator $U\colon L^2(\dR^d)\rightarrow L^2(\dR^d)$
such that the self-adjoint operators $-\Delta_{\delta,\alpha}$ and $-\Delta_{\delta',\beta}$ satisfy the operator inequality
\[
U^{-1}(-\Delta_{\delta',\beta})U \le -\Delta_{\delta,\alpha},
\]
and hence the assertions in Corollary~\ref{thm:main2} hold.
\end{cor}

The following example shows that Corollary~\ref{cor2} is sharp.

\begin{example}\label{exchi2}
Consider the Lipschitz partition $\cP=\{\dR^2_+,\dR^2_-\}$ of $\dR^2$ in the upper and lower half plane with boundary $\Sigma=\dR$. 
For constants $\alpha,\beta>0$ the spectra of the operators
$-\Delta_{\delta,\alpha}$ and $-\Delta_{\delta',\beta}$ can be computed via separation of variables; they are given by 
\[
\sigma(-\Delta_{\delta,\alpha}) = \sess(-\Delta_{\delta,\alpha})= [-\alpha^2/4,\infty)
\]
and
\[
\sigma(-\Delta_{\delta^\prime,\beta}) = \sess(-\Delta_{\delta^\prime,\beta})= [-4/\beta^2 ,\infty),
\]
respectively. Hence if  $\beta > 4/\alpha$ then 
\[
\min\sess(-\Delta_{\delta',\beta}) = -\frac{4}{\beta^2} > -\frac{\alpha^2}{4}=\min\sess(-\Delta_{\delta,\alpha})
\]
and it follows from Corollary~\ref{thm:main2}~(ii) that 
there exists no unitary operator $U$ in $L^2(\dR^2)$ for
which the operator inequality $U^{-1}(-\Delta_{\delta',\beta})U \le -\Delta_{\delta,\alpha}$ holds.
\end{example}

Another situation which is worth to mention is the case of a Lipschitz partition of $\dR^2$ which consists of a bounded domain and its complement, 
so that the chromatic number $\chi$ is again $2$.

\begin{example}
Consider the partition $\cP=\{\Omega,\dR^2\setminus\overline\Omega\}$, where $\Omega\subset\dR^2$ is a bounded domain with smooth boundary $\Sigma$, and let
 $\alpha,\beta>0$ be constant. In this case 
 $$\sess(-\Delta_{\delta,\alpha}) = \sess(-\Delta_{\delta',\beta}) = [0,\infty)$$ 
 and $N(-\Delta_{\delta,\alpha})\rightarrow +\infty$ as $\alpha \rightarrow +\infty$ according to \cite[Theorem 1]{EY02}. On the other hand
 we have $N(-\Delta_{\delta',\beta}) <\infty$ for any fixed $\beta > 0$ by \cite[Theorem 3.14\,(ii)]{BLL13}. Hence it follows from
 Corollary~\ref{thm:main2}\,(iii) and Theorem~\ref{thm:main} that for $\beta > 0$ 
there exists a sufficiently large $\alpha > 0$ such that the 
inequality $U^{-1}(-\Delta_{\delta',\beta})U \le -\Delta_{\delta,\alpha}$ 
fails for any unitary operator
of $L^2(\dR^2)$. 
\end{example}

In the next example, which forms a separate subsection, we discuss a particular situation with chromatic number $\chi=3$.

\subsection{An example: A symmetric star graph with three leads in $\dR^2$}\label{sec:chi3}

We consider a symmetric star graph in $\dR^2$ with three
leads such that any two leads form an angle
of degree $2\pi/3$, see Figure~\ref{stargraph}.

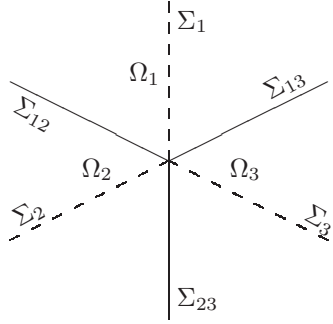
\begin{figure}[H]
\begin{center}
\begin{picture}(120,120)
\put(60,60){\line(2,1){60}}
\put(60,60){\line(-2,1){60}}
\put(60,60){\line(0,-1){60}}
\dashline{4}[0.7](60,60)(60,120)
\dashline{4}[0.7](60,60)(0,30)
\dashline{4}[0.7](60,60)(120,30)
\put(45, 90){$\Omega_1$}
\put(83, 54){$\Omega_3$}
\put(27, 54){$\Omega_2$}
\put(95, 83){\begin{turn}{29}$\Sigma_{13}$\end{turn}}
\put(0, 79){\begin{turn}{-29}$\Sigma_{12}$\end{turn}}
\put(0, 35){\begin{turn}{29}$\Sigma_{2}$\end{turn}}
\put(110, 38){\begin{turn}{-29}$\Sigma_{3}$\end{turn}}

\put(63, 5){$\Sigma_{23}$}
\put(63, 110){$\Sigma_{1}$}

\end{picture}

\end{center}
\caption{The star-graph $\Sigma = \Sigma_{12}\cup\Sigma_{23}\cup\Sigma_{13}$ separates
the Euclidean space $\dR^2$ into three congruent domains $\Omega_1$, $\Omega_2$ and $\Omega_3$ with bisector leads $\Sigma_1$, $\Sigma_2$ and $\Sigma_3$, respectively.} 
\label{stargraph}
\end{figure}

Let in the following $\alpha,\beta>0$ be real constants,
and let $-\Delta_{\delta,\alpha}$ and $-\Delta_{\delta',\beta}$ be the corresponding self-adjoint operators
with $\delta$ and $\delta^\prime$-interactions, respectively, supported on the star graph. Then we have
\begin{equation}\label{mindelta}
\min\sigma(-\Delta_{\delta,\alpha}) = - \frac{\alpha^2}{3}
\end{equation}
and
\begin{equation}\label{mindeltaprime}
\min\sigma(-\Delta_{\delta',\beta}) 
\ge - \Bigg(\frac{12\sqrt{3}-2}{9}\Bigg)^2\frac{1}{\beta^2}.
\end{equation}
Whereas \eqref{mindelta} is essentially a consequence of \cite[Lemma 2.6]{LP08} (and can be viewed as a strengthening of 
\cite[Theorem 3.2]{BEW09} in the present situation) the proof of \eqref{mindeltaprime}
is of more computational nature. Both proofs are outsourced in the appendix.

Clearly the chromatic number of the partition  of $\dR^2$ corresponding to the star graph in Figure~\ref{stargraph}
is $\chi=3$ and hence the operator inequality 
\begin{equation*}
U^{-1}(-\Delta_{\delta',\beta})U \le -\Delta_{\delta,\alpha}
\end{equation*}
for the corresponding Laplacians 
in Theorem~\ref{thm:main} is valid under the condition
\begin{equation}\label{easter1}
0\leq\beta\leq \frac{4}{\alpha}\sin^2\bigl(\pi/3\bigr)=\frac{3}{\alpha}.
\end{equation}
We point out that the assumption \eqref{easter1} can not be replaced by the weaker
assumption 
\begin{equation*}
0\leq\beta\leq \frac{4}{\alpha}\sin^2\bigl(\pi/2\bigr)=\frac{4}{\alpha},
\end{equation*}
which corresponds to the case $\chi=2$ in Theorem~\ref{thm:main}. In fact, for $c^*:=4-\frac{2\sqrt{3}}{9}$ we have $3<c^*<4$, and if we 
choose $\alpha,\beta>0$ such that
$\beta > c^* \frac{1}{\alpha}$ then we conclude $\min\sigma(-\Delta_{\delta,\alpha})<\min\sigma(-\Delta_{\delta',\beta})$ 
from \eqref{mindelta} and \eqref{mindeltaprime}. This yields the following corollary.

\begin{cor}\label{cor:BEL}
Let $\alpha,\beta>0$ and $
\beta > c^*\frac{1}{\alpha}$, where $c^*=4-\frac{2\sqrt{3}}{9}$.  Then there exists no unitary operator $U$ in $L^2(\dR^2)$
such that 
\[
U^{-1}(-\Delta_{\delta',\beta})U \le -\Delta_{\delta,\alpha}.
\] 
\end{cor}

\section{Essential spectra and bound states of Schr\"{o}dinger operators with $\delta$ and $\delta^\prime$-interactions}\label{sec4}

In this section we discuss some spectral properties of the Schr\"{o}dinger operators $-\Delta_{\delta,\alpha}$ and $-\Delta_{\delta',\beta}$, 
where the $\delta$ and $\delta^\prime$-interaction, respectively, is
supported on certain Lipschitz partitions of $\dR^d$. We are mainly interested in the following two situations: 
Lipschitz partitions with compact boundaries in Section~\ref{4.1} and Lipschitz partitions 
which are deformed on a compact subset of $\dR^d$ in Section~\ref{sec:ex2}.
Special attention is paid to bound states in the cases $d=2$ and $d=3$ in Section~\ref{sec23case}.

\subsection{Lipschitz partitions with compact boundary}\label{4.1}

Throughout this subsection we assume that the following hypothesis is satisfied.

\begin{hyp}\label{hyp:compact}
Let $\cP = \{\Omega_k\}_{k=1}^n$ be a Lipschitz partition of $\dR^d$, $d\ge 2$, such that  
the boundary $\Sigma=\cup_{k=1}^n\partial\Omega_k$ is compact.
\end{hyp}

By Hypothesis~\ref{hyp:compact} the partition $\cP = \{\Omega_k\}_{k=1}^n$ consists of $(n-1)$ bounded domains and one unbounded domain;
cf. Figure~\ref{fig:compact}. We shall call these type of partitions sometimes \emph{compact Lipschitz partitions}. 

\begin{figure}[H]
\begin{center}
\begin{picture}(170,190)(-20,0)
\qbezier(5,130)(55,190)(105,140)
\qbezier(105,140)(134,100)(85,80)
\qbezier(85,80)(55,60)(27,100)
\qbezier(27,100)(15,110)(15,110)
\qbezier(5,130)(-5,115)(15,110)
\qbezier(5,130)(35,100)(105,140)
\put(45,100){$\Omega_{1}$}
\put(45,130){$\Omega_{2}$}
\put(5,80){$\Omega_3$}
\put(5,25){$\cP = \{\Omega_k\}_{k=1}^3$,~ $\,\,\chi = 3$}
\put(-22,179){$\dR^2$}
\qbezier(-25,190)(95,190)(145,190)
\qbezier(-25,40)(95,40)(145,40)
\qbezier(-25,40)(-25,130)(-25,190)
\qbezier(145,40)(145,140)(145,190)
\end{picture}
\begin{picture}(170,190)(-30,0)
\qbezier(5,130)(55,190)(105,140)
\qbezier(105,140)(134,100)(85,80)
\qbezier(85,80)(55,60)(27,100)
\qbezier(27,100)(15,110)(15,110)
\qbezier(5,130)(-5,115)(15,110)
\qbezier(5,130)(35,100)(105,140)
\qbezier(55,118)(65,100)(55,74.5)
\put(20,25){$\cP = \{\Omega_k\}_{k=1}^4$,~ $\,\,\chi = 4$}
\put(35,100){$\Omega_{1}$}
\put(45,130){$\Omega_{2}$}
\put(5,80){$\Omega_3$}
\put(90,100){$\Omega_{4}$}
\put(-22,179){$\dR^2$}
\qbezier(-25,190)(95,190)(145,190)
\qbezier(-25,40)(95,40)(145,40)
\qbezier(-25,40)(-25,130)(-25,190)
\qbezier(145,40)(145,140)(145,190)

\end{picture}
\end{center}
\caption{Examples of compact Lipschitz partitions with chromatic numbers $3$ and $4$.}
\label{fig:compact}
\end{figure}
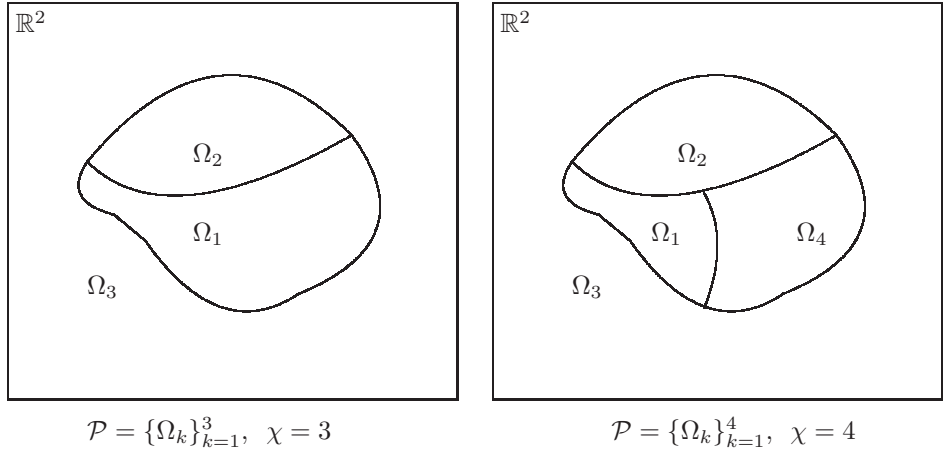

In the next theorem we show that under Hypothesis~\ref{hyp:compact} the operators $-\Delta_{\delta,\alpha}$ and $-\Delta_{\delta',\beta}$
are compact perturbations of the free Laplacian $-\Delta_{\rm free}$ defined on $H^2(\dR^d)$. A variant of Theorem~\ref{resdiff1}~(i) is also 
contained in \cite[Theorem 4]{H89} and in \cite[Theorem 3.1]{BEKS94}; cf. \cite{BELPAMM} for a detailed proof in the present situation.
We also mention that for a compact partition consisting of $C^\infty$-smooth domains it can be shown that 
the resolvent differences below belong to certain Schatten-von Neumann ideals depending on the space dimension $d$. We refer the reader
to \cite{BLL13} for more details.

\begin{thm}\label{resdiff1}
Let $\cP = \{\Omega_k\}_{k=1}^n$ be a compact Lipschitz partition of $\dR^d$ with boundary $\Sigma$ as in Hypothesis~\ref{hyp:compact}, let 
$\alpha,\beta:\Sigma\rightarrow\dR$ be such that $\alpha,\beta^{-1}\in L^\infty(\Sigma)$, and let $-\Delta_{\delta,\alpha}$ and
$-\Delta_{\delta',\beta}$ be the self-adjoint operators associated with $\cP$. Then the following statements hold.
\begin{itemize}
 \item [{\rm (i)}] For all $\lambda\in\rho(-\Delta_{\rm free})\cap\rho(-\Delta_{\delta,\alpha})$ the resolvent difference 
 \begin{equation*}
(-\Delta_{\rm free} -\lambda)^{-1} - (-\Delta_{\delta,\alpha}-\lambda)^{-1}
\end{equation*}
is a compact operator in $L^2(\dR^d)$.
\item [{\rm (ii)}] For all $\lambda\in\rho(-\Delta_{\rm free})\cap\rho(-\Delta_{\delta^\prime,\beta})$ the resolvent difference 
 \begin{equation*}
(-\Delta_{\rm free} -\lambda)^{-1} - (-\Delta_{\delta^\prime,\beta}-\lambda)^{-1}
\end{equation*}
is a compact operator in $L^2(\dR^d)$.
\end{itemize}
In particular, $\sess(-\Delta_{\delta,\alpha}) = \sess(-\Delta_{\delta',\beta}) = [0,\infty)$.
\end{thm}

\begin{proof}
We shall only prove item (ii). 
The proof of item (i) is along the same lines and can also be found in the note \cite{BELPAMM}.
Let us fix $\lambda_0<\min \sigma(-\Delta_{\delta',\beta})$ and 
set
$$
W:=(-\Delta_{\rm free} -\lambda_0)^{-1} - (-\Delta_{\delta^\prime,\beta}-\lambda_0)^{-1}.
$$
For $f, g\in L^2(\dR^d)$ we define the functions 
\begin{equation*}
u := (-\Delta_{\rm free}-\lambda_0 )^{-1}f \quad\text{and} \quad v:= (-\Delta_{\delta',\beta}-\lambda_0)^{-1}g.
\end{equation*} 
Then we compute
\begin{equation}\label{Wfg3}
\begin{split}
(W f, g)_{L^2(\dR^d)} &= \bigl((-\Delta_{\rm free}-\lambda_0 )^{-1}f, g\bigr)_{L^2(\dR^d)}-\bigl(f,(-\Delta_{\delta',\beta}-\lambda_0)^{-1}g\bigr)_{L^2(\dR^d)} \\
&= \bigl(u,(-\Delta_{\delta',\beta}-\lambda_0)v \bigr)_{L^2(\dR^d)}-\bigl((-\Delta_{\rm free} -\lambda_0)u,v\bigr)_{L^2(\dR^d)} \\
&=(u,-\Delta_{\rm \delta',\beta}v)_{L^2(\dR^d)} - (-\Delta_{\rm free}u,v)_{L^2(\dR^d)}.
\end{split}
\end{equation}
Observe that $u\in H^2(\dR^d)\subset\dom\fra_{\delta',\beta}$ and that for any common boundary $\Sigma_{kl}$ 
with $k,l=1,2,\dots,n$ and $k\ne l$ the condition
 $u_k|_{\Sigma_{kl}} = u_l|_{\Sigma_{kl}}$ holds. Hence we have
\begin{equation}\label{delta'uv}
(u,-\Delta_{\delta',\beta}v)_{L^2(\dR^d)} =
\sum_{k=1}^n\big(\nabla u_k,\nabla v_k\big)_{L^2(\Omega_k;\dC^d)}, 
\end{equation}
where we used the definition of $\fra_{\delta',\beta}$ from \eqref{delta'form}. Furthermore,
we obtain with the help of Green's first identity (see e.g. \cite[Lemma 4.1]{McLean})
\begin{equation}\label{freeuv}
(-\Delta_{\rm free}u,v)_{L^2(\dR^d)} = \sum_{k=1}^n\big(\nabla u_k,\nabla v_k\big)_{L^2(\Omega_k;\dC^d)}  - 
\sum_{k=1}^n (\partial_{\nu_k} u_k|_{\partial\Omega_k}, v_k|_{\partial\Omega_k})_{L^2(\partial\Omega_k)};
\end{equation} 
here we 
also used that the restrictions $u_k$, $v_k$ satisfy $u_k\in H^2(\Omega_k)$, 
$v_k\in H^1(\Omega_k)$ and, hence, 
$\partial_{\nu_k}u_k|_{\partial\Omega_k}, v_k|_{\partial\Omega_k}  \in H^{1/2}(\partial\Omega_k)\subset L^2(\partial\Omega_k)$.
Combining \eqref{Wfg3} with \eqref{delta'uv} and \eqref{freeuv} we obtain
\begin{equation*}
\begin{split}
\big(Wf,g\big)_{L^2(\dR^d)} &= \sum_{k=1}^n \big(\partial_{\nu_k} u_k|_{\partial\Omega_k}, v_k|_{\partial\Omega_k}\big)_{L^2(\partial\Omega_k)}.
\end{split}
\end{equation*}

Let $\cG := \bigoplus_{k=1}^n L^2(\partial\Omega_k)$ and $\cG^{1/2} := \bigoplus_{k=1}^n H^{1/2}(\partial\Omega_k)$, and define the operators
$T_1,T_2\colon L^2(\dR^d)\rightarrow \cG$ by
\begin{equation*}
 \begin{split}
  T_1f :=& \bigoplus_{k=1}^n \partial_{\nu_k}\big[(-\Delta_{\rm free} -\lambda_0)^{-1}f\big]_k\big|_{\partial\Omega_k}
= \bigoplus_{k=1}^n \partial_{\nu_k} u_k |_{\partial\Omega_k}, \\
T_2g := & \bigoplus_{k=1}^n\big[(-\Delta_{\rm \delta',\beta} -\lambda_0)^{-1}g\big]_k\big|_{\partial\Omega_k}= 
\bigoplus_{k=1}^n v_k |_{\partial\Omega_k}.
 \end{split}
\end{equation*}
As $(-\Delta_{\rm free} -\lambda_0)^{-1}$ is continuous from $L^2(\dR^d)$ into $H^2(\dR^d)$ and $(-\Delta_{\rm \delta',\beta} -\lambda_0)^{-1}$
is continuous from $L^2(\dR^d)$ into $\dom \mathfrak a_{\delta^\prime,\beta}$ it follows from the continuity of the trace maps 
that both operators $T_1$ and $T_2$ are continuous from $L^2(\dR^d)$ into $\cG^{1/2}$; cf. \cite[Theorem 3.37]{McLean}.
Since $\cG^{1/2}$ is compactly embedded in $\cG$ both operators $T_1,T_2\colon L^2(\dR^d)\rightarrow \cG$ are compact. 
From $(Wf, g)_{L^2(\dR^d)} = \big(T_1f, T_2g\big)_{\cG}$
we conclude that
$$T_2T_1^*=W= (-\Delta_{\rm free} -\lambda_0)^{-1} - (-\Delta_{\delta^\prime,\beta}-\lambda_0)^{-1}$$ is a compact operator in $L^2(\dR^d)$. 
Now a standard argument shows that the resolvent difference is compact for all $\lambda\in\rho(-\Delta_{\rm free})\cap\rho(-\Delta_{\delta',\beta})$,
see, e.g., \cite[Lemma~2.2]{BLL12a}.

Finally, note that $\sigma(-\Delta_{\rm free}) =  \sess(-\Delta_{\rm free}) = [0,\infty)$ and hence the assertion on the
essential spectra of $-\Delta_{\delta,\alpha}$ and
$-\Delta_{\delta',\beta}$ follows from the compactness of the resolvent differences in (i) and (ii).
\end{proof}

The next statement on the negative eigenvalues of $-\Delta_{\delta,\alpha}$ and $-\Delta_{\delta',\beta}$
is an immediate consequence of Corollary~\ref{thm:main2} and the fact that the essential spectra of 
$-\Delta_{\delta,\alpha}$ and $-\Delta_{\delta',\beta}$ coincide.

\begin{cor}
Let the assumptions be as in Theorem~\ref{resdiff1} and assume, in addition, that
\[
0 < \beta \le  \frac{4}{\alpha}\sin^2\big(\pi / \chi\big),
\]
where $\chi$ is the chromatic number of the partition $\cP$.
Let $\{\lambda_{k}(-\Delta_{\delta,\alpha})\}_{k=1}^{\infty}$ and
$\{\lambda_{k}(-\Delta_{\delta',\beta})\}_{k=1}^{\infty}$ be the negative eigenvalues of $-\Delta_{\delta,\alpha}$ and $-\Delta_{\delta',\beta}$,
respectively, 
and let $N(-\Delta_{\delta,\alpha})$ and $N(-\Delta_{\delta',\beta})$ be their total multiplicities as in Definition~\ref{dfn:spec1}.
Then the following statements hold:
\begin{itemize}\setlength{\itemsep}{1.2ex}
\item [\rm (i)] $\lambda_k(-\Delta_{\delta',\beta})\le \lambda_k(-\Delta_{\delta,\alpha})$ for all $k\in\dN$;
\item[\rm (ii)] $N(-\Delta_{\delta,\alpha})\leq N(-\Delta_{\delta',\beta})$.
\end{itemize}
\end{cor}

Finally we show that the Schr\"o\-dinger operator with a $\delta'$-interaction
of strength $\beta>0$ has at least one negative eigenvalue.

\begin{thm}
\label{boundstate1}
Let $\cP = \{\Omega_k\}_{k=1}^n$ be a compact Lipschitz partition of $\dR^d$ with boundary $\Sigma$ as in Hypothesis~\ref{hyp:compact}, let 
$\beta^{-1}\in L^\infty(\Sigma)$ be real, and let 
$-\Delta_{\delta',\beta}$ be the self-adjoint operator with $\delta^\prime$-interaction supported on $\Sigma$. If 
\[
\int_{\partial\Omega_k} \beta^{-1}(x)d\sigma_k(x) > 0
\]
holds for some bounded $\Omega_k$, $k\in 1,\dots,n$, then $N(-\Delta_{\delta',\beta}) \ge 1$. In particular, if $\beta>0$ is a real constant then
$-\Delta_{\delta',\beta}$ has at least one negative eigenvalue. 
\end{thm}

\begin{proof}
Let
$f=\chi_{\Omega_k}$ be the characteristic function of $\Omega_k$. Then  $f\in \dom \fra_{\delta',\beta}$, $\nabla f =0$, and hence
\[
\fra_{\delta',\beta}[f] = - \int_{\partial\Omega_k}\beta^{-1}(x)d\sigma(x) < 0.
\]
This implies
$\inf\sigma(-\Delta_{\delta',\beta}) < 0$. 
\end{proof}

\begin{remark}
There is no general analog of Theorem~\ref{boundstate1} for $\delta$-interactions.
In space dimensions $d \ge 3$ it follows implicitly from the Birman-Schwinger-type estimate in
\cite[Theorem 4.2\,(iii)]{BEKS94} that for $\|\alpha\|_\infty$ sufficiently
small the operator $-\Delta_{\delta,\alpha}$ has no bound states.
The existence of eigenvalues
   depends not only on $\alpha$, but also on the geometry of
   the support of the interaction; an example in the case $d=3$ is
   discussed in \cite{EF09}.
The picture is different in space dimension $d = 2$. 
In the simple case of a constant strength
$\alpha >0$ along the support of the  interaction at least one bound state always exists, see \cite{ET04}.
\end{remark}

\subsection{Locally deformed partitions of $\dR^d$}\label{sec:ex2}

In this section we consider non-compact partitions consisting of finitely many Lipschitz domains.
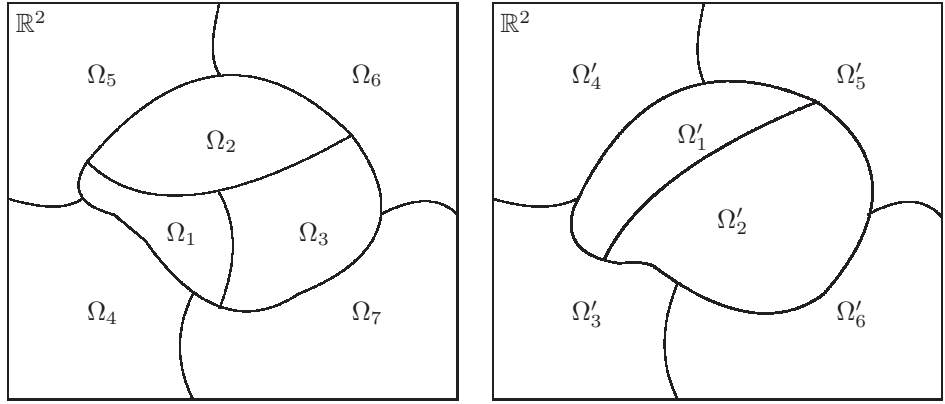
\begin{figure}[H]
\begin{center}
\begin{picture}(170,190)(-20,0)
\qbezier(5,130)(55,190)(105,140)
\qbezier(105,140)(135,100)(85,80)
\qbezier(85,80)(55,60)(27,100)
\qbezier(27,100)(15,110)(15,110)
\qbezier(5,130)(-5,115)(15,110)
\qbezier(5,130)(35,100)(105,140)
\qbezier(55,118)(65,100)(55,74.5)

\qbezier(-25,190)(95,190)(145,190)
\qbezier(-25,40)(95,40)(145,40)
\qbezier(-25,40)(-25,130)(-25,190)
\qbezier(145,40)(145,140)(145,190)


\qbezier(55,163)(50,170)(55,190)
\qbezier(3,116)(-5,110)(-25,116)
\qbezier(45,80)(35,60)(45,40)
\qbezier(116,110)(135,120)(145,110)
\put(35,100){$\Omega_{1}$}
\put(85,100){$\Omega_{3}$}
\put(50,135){$\Omega_{2}$}
\put(5,70){$\Omega_4$}
\put(5,160){$\Omega_5$}
\put(105,70){$\Omega_7$}
\put(105,160){$\Omega_6$}
\put(-22,179){$\dR^2$}
\put(5,15){$\cP = \{\Omega_k\}_{k=1}^7$,~ $\chi = 4$}
\end{picture}
\begin{picture}(170,190)(-30,0)
\begin{rotate}{18}
\qbezier(40,105)(90,160)(140,115)
\qbezier(140,115)(170,75)(120,45)
\qbezier(120,45)(90,35)(62,75)
\qbezier(62,75)(59,78)(50,80)
\qbezier(40,105)(30,90)(50,80)
\qbezier(45,83)(70,110)(140,115)
\end{rotate}

\qbezier(-25,190)(95,190)(145,190)
\qbezier(-25,40)(95,40)(145,40)
\qbezier(-25,40)(-25,130)(-25,190)
\qbezier(145,40)(145,140)(145,190)


\qbezier(55,160)(50,170)(55,190)
\qbezier(7,116)(-5,110)(-25,116)
\qbezier(45,84)(35,60)(45,40)
\qbezier(117,110)(135,120)(145,109)
\put(60,105){$\Omega_{2}'$}
\put(45,137){$\Omega_{1}'$}
\put(5,70){$\Omega_3'$}
\put(5,160){$\Omega_4'$}
\put(105,70){$\Omega_6'$}
\put(105,160){$\Omega_5'$}
\put(-22,179){$\dR^2$}
\put(5,15){$\cP' = \{\Omega_k'\}_{k=1}^6$,~ $\chi = 3$}
\end{picture}
\end{center}
\caption{A non-compact Lipschitz partition $\cP = \{\Omega_k\}_{k=1}^7$ of $\dR^2$ with chromatic number $\chi=4$ 
and a local deformation $\cP' = \{\Omega_k'\}_{k=1}^6$ with chromatic number $\chi = 3$.}
\label{fig:noncompact}
\end{figure}

Let in the following $\cP = \{\Omega_k\}_{k=1}^n$ and $\cP^\prime = \{\Omega_k^\prime\}_{k=1}^{n^\prime}$ be 
Lipschitz partitions of $\dR^d$ with boundaries $\Sigma$ 
and $\Sigma^\prime$, respectively. We say that $\cP$ and $\cP^\prime$ are {\it local deformations of each other} if there
exists a bounded domain $\cB$ such that
\begin{equation}\label{bsigma}
\Sigma \setminus \cB = \Sigma'\setminus \cB,
\end{equation}
see Figure~\ref{fig:noncompact}.
In addition it will be assumed that there exist $C^{1,1}$ components in the boundary $\Sigma$ (and $\Sigma^\prime$) and that
$\cB$ can be chosen in such a way that $\partial\cB\cap\Sigma$ is contained in these components. The following hypothesis
makes this more precise.

\begin{hyp}\label{hypolocal}
Let $\cP = \{\Omega_k\}_{k=1}^n$ and $\cP^\prime = \{\Omega_k^\prime\}_{k=1}^{n^\prime}$ be locally deformed Lipschitz partitions and
let $\cB$ be a bounded domain with smooth boundary $\partial\cB$ such that \eqref{bsigma} holds.
Let  $\cB_0$ and $\cB_1$ be bounded domains  such that $\overline \cB_0\subset \cB$, $\overline\cB\subset\cB_1$, and assume that
$$\Gamma:=\bigl(\cB_1\setminus\overline\cB_0\bigr) \,\cap\,\Sigma=\bigl(\cB_1\setminus\overline\cB_0\bigr) \,\cap\,\Sigma^\prime$$
consists of $C^{1,1}$ components of a Lipschitz dissection of $\Sigma$, or equivalently, of $\Sigma^\prime$.
\end{hyp}

In the next theorem we prove that the essential spectra of the Schr\"{o}dinger operators $-\Delta_{\delta,\alpha}$ and 
$-\Delta_{\delta^\prime,\beta}$ do not change under local deformations of Lipschitz partitions. Our proof is partly inspired
by \cite[Theorem 6.1 in English translation]{B62}, where similar arguments were used 
for elliptic operators with Robin and mixed boundary conditions
under local deformations of the boundary and local variations of the Robin coefficient.

\begin{thm}\label{thm:compact}
Let $\cP = \{\Omega_k\}_{k=1}^{n}$ and $\cP^\prime = \{\Omega_k^\prime\}_{k=1}^{n^\prime}$ be Lipschitz partitions of $\dR^d$ which are 
local deformations of each other such that Hypothesis~\ref{hypolocal} holds. 
Let $\alpha,\beta^{-1} \in L^\infty(\Sigma)$ and $\alpha^\prime,\beta^{\prime -1} \in L^\infty(\Sigma^\prime)$ be real and assume that 
\[
\alpha|_{\Sigma\setminus \cB_0}  = \alpha^\prime|_{\Sigma^\prime\setminus \cB_0},\qquad 
\beta|_{\Sigma\setminus \cB_0}  = \beta^\prime|_{\Sigma^\prime\setminus \cB_0}\qquad\text{and}\qquad 
\alpha|_{\Gamma}, \beta^{-1}|_{\Gamma}\in C^1(\Gamma).
\]
Let $-\Delta_{\delta,\alpha}$, $-\Delta_{\delta',\beta}$, and $-\Delta^\prime_{\delta,\alpha^\prime}$, $-\Delta^\prime_{\delta',\beta^\prime}$  
be the Schr\"odinger associated with the partitions $\cP$ and $\cP^\prime$, respectively.
Then the following statements hold.
\begin{itemize}
\item [\rm (i)] For all $\lambda\in\rho(-\Delta_{\delta,\alpha})\cap\rho(-\Delta^\prime_{\delta,\alpha^\prime})$ the resolvent difference
\[
(-\Delta_{\delta,\alpha} - \lambda)^{-1} - (-\Delta^\prime_{\delta,\alpha^\prime}-\lambda)^{-1}
\]
is a compact operator in $L^2(\dR^d)$. In particular, $\sigma_{\rm ess}(-\Delta_{\delta,\alpha}) = \sigma_{\rm ess}(-\Delta^\prime_{\delta,\alpha^\prime})$.
\item [\rm (ii)]
For all $\lambda\in\rho(-\Delta_{\delta',\beta})\cap\rho(-\Delta^\prime_{\delta',\beta^\prime})$ the resolvent difference
\[
(-\Delta_{\delta',\beta}- \lambda)^{-1} - (-\Delta^\prime_{\delta',\beta^\prime}-\lambda)^{-1}
\]
is a compact operator in $L^2(\dR^d)$. In particular, $\sigma_{\rm ess}(-\Delta_{\delta',\beta}) = \sigma_{\rm ess}(-\Delta^\prime_{\delta',\beta^\prime})$.
\end{itemize}
\end{thm}

\begin{proof}
The proof of Theorem~\ref{thm:compact} will be given only for the simple case that both 
Lipschitz partitions consist of two domains only, that is, $n=n^\prime=2$. The general case requires more notation but follows the 
same strategy.
We verify (ii), the proof of (i) is similar. The fact that the essential spectra of $-\Delta_{\delta,\alpha}$ and $-\Delta^\prime_{\delta,\alpha^\prime}$,
and $-\Delta_{\delta',\beta}$ and $-\Delta^\prime_{\delta',\beta^\prime}$ coincide is a direct consequence of the compactness
of their resolvent differences in (i) and (ii).

Let us fix some notation; cf. Figure~\ref{fig3}. Set
$$
\Omega_{i1}:=\Omega_i\cap\cB,\quad \Omega_{i2}:=\Omega_i\cap(\dR^d\setminus\overline\cB),\qquad i=1,2,
$$
denote the restrictions of functions $f_k$ on $\Omega_k$ onto $\Omega_{kl}$ by $f_{kl}$, $k,l=1,2$, and let
$$
\Sigma_1:=\Sigma\cap\cB,\quad \Sigma_2:=\Sigma\cap(\dR^d\setminus\overline\cB).$$ 

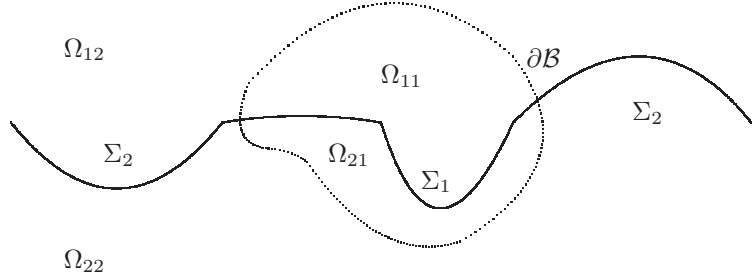
\begin{figure}[H]
\begin{center}
\begin{picture}(300,150)
\qbezier(10,75)(50,25)(90,75)
\qbezier(200,75)(250,125)(290,75)
\qbezier(90,75)(120,80)(150,75)
\qbezier(150,75)(170,10)(200,75)
\bezier{60}(100,90)(150,145)(200,100)
\bezier{50}(200,100)(230,60)(180,30)
\bezier{50}(180,30)(150,20)(122,60)
\bezier{10}(122,60)(119,63)(110,65)
\bezier{30}(100,90)(90,65)(110,65)
\put(30,100){$\Omega_{12}$}
\put(30,20){$\Omega_{22}$}
\put(130,60){$\Omega_{21}$}
\put(150,90){$\Omega_{11}$}
\put(45,60){$\Sigma_2$}
\put(165,50){$\Sigma_1$}
\put(245,76){$\Sigma_2$}
\put(205,96){$\partial\cB$}
\end{picture}
\end{center}
\caption{The hypersurface $\partial\cB$ splits the domain $\Omega_1$ into the parts $\Omega_{11}$ and $\Omega_{12}$, 
and the domain $\Omega_2$ into the parts $\Omega_{21}$ and $\Omega_{22}$. The hypersurface $\Sigma$ splits into $\Sigma_1$ and $\Sigma_2$.}
\label{fig3}
\end{figure}
\noindent We denote the restriction of $\beta$ onto $\Sigma_i$ by $\beta_i$, $i=1,2$. 
In the present situation the sesquilinear form $\mathfrak a_{\delta^\prime,\beta}$ in \eqref{delta'form} is given by 
\begin{equation*}
 \fra_{\delta',\beta}[f,g] = \sum_{k=1}^2 \big(\nabla f_k,\nabla g_k\big)_{L^2(\Omega_k;\dC^d)}-
\big(\beta^{-1}(f_1|_\Sigma - f_2|_\Sigma), g_1|_\Sigma - g_2|_\Sigma\big)_{L^2(\Sigma)}
\end{equation*}
with $\dom \fra_{\delta',\beta} =  H^1(\Omega_1)\oplus H^1(\Omega_2)$. Observe that the right hand side can
also be written in the form
\begin{equation*}
\sum_{k,l=1}^{2} (\nabla f_{kl}, \nabla g_{kl})_{L^2(\Omega_{kl};\dC^d)} 
 -\sum_{j=1}^2\big(\beta_j^{-1}(f_{1j}|_{\Sigma_j} -f_{2j}|_{\Sigma_j}),
g_{1j}|_{\Sigma_j} -g_{2j}|_{\Sigma_j}\bigr)_{L^2(\Sigma_j)}.
\end{equation*}

\noindent {\bf Step I.}
We introduce an auxiliary sesquilinear form by
\begin{equation*}
\begin{split}
\fra_{\delta',\beta,\rm N}[f,g]\! &:= \sum_{k,l=1}^{2} (\nabla f_{kl}, \nabla g_{kl})_{L^2(\Omega_{kl};\dC^d)} \\
& \qquad -\sum_{j=1}^2\big(\beta_j^{-1}(f_{1j}|_{\Sigma_j} -f_{2j}|_{\Sigma_j}),
g_{1j}|_{\Sigma_j} -g_{2j}|_{\Sigma_j}\bigr)_{L^2(\Sigma_j)},\\
\dom\fra_{\delta',\beta,\rm N} &= \bigoplus_{k,l=1}^{2} H^1(\Omega_{kl}).
\end{split}
\end{equation*}
As in the proof of Proposition~\ref{prop:deltaform} one verifies that $\fra_{\delta',\beta,\rm N}$ 
is a closed, densely defined form which is semibounded from below, and hence gives rise to a self-adjoint operator 
$-\Delta_{\delta',\beta,\rm N}$ in $L^2(\dR^d)$. Note that the functions in the domain of $-\Delta_{\delta',\beta,\rm N}$ satisfy
Neumann boundary conditions on $\partial\cB\cap\Omega_i$, $i=1,2$, and the same $\delta^\prime$-type boundary conditions at $\Sigma_i$, $i=1,2$, 
as the functions in the domain of $-\Delta_{\delta',\beta}$. In this step we show that 
\begin{equation}
\label{W3}
(-\Delta_{\delta',\beta} - \lambda)^{-1} - (-\Delta_{\delta',\beta,\rm N} -\lambda)^{-1}
\end{equation}
is a compact operator in $L^2(\dR^d)$ for all $\lambda\in\rho(-\Delta_{\delta',\beta})\cap\rho(-\Delta_{\delta',\beta,\rm N})$.

In fact, choose $\lambda_0 < \min\{\min\sigma(-\Delta_{\delta',\beta}),
\min\sigma(-\Delta_{\delta',\beta, \rm N})\}$ and let 
$W$ be the resolvent difference in \eqref{W3} with $\lambda = \lambda_0$.
For $f, g\in L^2(\dR^d)$ define 
\begin{equation*}
u := (-\Delta_{\delta',\beta} -\lambda_0)^{-1}f\quad\text{and} \quad v:= (-\Delta_{\delta',\beta,\rm N}-\lambda_0)^{-1}g.
\end{equation*} 
A straightforward computation as in \eqref{Wfg3} yields 
\begin{equation}
\label{Wfg4}
(Wf, g)_{L^2(\dR^d)}= (u,-\Delta_{\delta',\beta, \rm N}v)_{L^2(\dR^d)} - (-\Delta_{\delta',\beta}u,v)_{L^2(\dR^d)}.
\end{equation}
As $u\in\dom(-\Delta_{\delta',\beta})\subset \dom\fra_{\delta',\beta} \subset \dom\fra_{\delta',\beta,\rm N}$ we have
for the first term on the right hand side
\[
\begin{split}
&(u,-\Delta_{\delta',\beta, \rm N}v)_{L^2(\dR^d)}\\ 
&\quad=\sum_{k,l=1}^{2}  (\nabla u_{kl}, \nabla v_{kl})_{L^2(\Omega_{kl};\dC^d)}
 - \sum_{j=1}^2\bigl(\beta^{-1}_j(u_{1j}|_{\Sigma_j}-u_{2j}|_{\Sigma_j}),
v_{1j}|_{\Sigma_j}-v_{2j}|_{\Sigma_j}\bigr)_{L^2(\Sigma_j)}.
\end{split}
\]
In order to rewrite the second term on the right hand side of \eqref{Wfg4} note first that for $u\in\dom(-\Delta_{\delta',\beta})$
we have 
$$
\partial\nu_{j1}u_{j1}\vert_{\partial\cB\cap\Omega_j}+\partial\nu_{j2}u_{j2}\vert_{\partial\cB\cap\Omega_j}=0,\qquad j=1,2;
$$
here the Neumann traces exist in $H^{1/2}(\partial\cB\cap\Omega_j)$ due to the $H^2$-regularity of the functions in 
$\dom(-\Delta_{\delta^\prime,\beta})$ near $\partial\cB\cap\Omega_j$ (which follows from $u_j\in H^2_{\rm loc}(\Omega_j)$ and 
Lemma~\ref{h2bound}~(ii)). Moreover 
$u\in \dom(-\Delta_{\delta',\beta})$ satisfies the boundary conditions
$$
\partial\nu_{1j}u_{1j}\vert_{\Sigma_j}=\beta_j^{-1}(u_{1j}\vert_{\Sigma_j}-u_{2j}\vert_{\Sigma_j})=-\partial\nu_{2j}u_{2j}\vert_{\Sigma_j}
,\qquad j=1,2,
$$
by Theorem~\ref{thm:delta} (ii)-(c$^\prime$).
Hence we obtain for the second term on the right hand side of \eqref{Wfg4} when integrating by parts,
\[
\begin{split}
&(-\Delta_{\delta',\beta}u,v)_{L^2(\dR^d)}\\
&\quad = \sum_{k,l=1}^{2}  (\nabla u_{kl}, \nabla v_{kl})_{L^2(\Omega_{kl};\dC^d)}- \sum_{j=1}^2\bigl(\beta_j^{-1}(u_{1j}|_{\Sigma_j}-u_{2j}|_{\Sigma_j}),
v_{1j}|_{\Sigma_j}-v_{2j}|_{\Sigma_j}\bigr)_{L^2(\Sigma_{j})}\\
&\qquad\qquad-\sum_{j=1}^2\big(\partial_{\nu_{j1}}u_{j1}|_{\partial \cB\cap\Omega_j},v_{j1}|_{\partial \cB\cap\Omega_j} 
- v_{j2}|_{\partial \cB\cap\Omega_j}\big)_{L^2(\partial \cB\cap\Omega_j)}.
\end{split}
\]
Thus \eqref{Wfg4} has the form
\begin{equation}\label{Wfg5}
\begin{split}
(Wf,g)_{L^2(\dR^d)} & = \sum_{j=1}^2\big(\partial_{\nu_{j1}}u_{j1}|_{\partial \cB\cap\Omega_j},v_{j1}|_{\partial \cB\cap\Omega_j} - 
v_{j2}|_{\partial \cB\cap\Omega_j}\big)_{L^2(\partial \cB\cap\Omega_j)}\\
&= (T_1 f, T_2g)_{L^2(\partial \cB)},
\end{split}
\end{equation}
where the operators $T_1, T_2\colon L^2(\dR^d)\rightarrow L^2(\partial \cB)$ are defined by
\[
\begin{split}
T_1f :=& \bigoplus_{j=1}^2 \partial_{\nu_{j1}}\bigl[(-\Delta_{\delta',\beta} -\lambda_0)^{-1}f\bigr]_{j1}\bigl|_{\partial \cB \cap\Omega_j}=
          \bigoplus_{j=1}^2 \partial_{\nu_{j1}} u_{j1}|_{\partial \cB\cap\Omega_j}                       \\
T_2g :=& \bigoplus_{j=1}^2\Big[\bigl[(-\Delta_{\delta',\beta,\rm N} -\lambda_0)^{-1}g\bigr]_{j1}|_{\partial \cB\cap\Omega_j}-
\bigl[(-\Delta_{\delta',\beta,\rm N} -\lambda_0)^{-1}g\bigr]_{j2}|_{\partial \cB\cap\Omega_j}\Big]\\
 =& \bigoplus_{j=1}^2 \bigl[v_{j1}|_{\partial \cB\cap\Omega_j} - 
v_{j2}|_{\partial \cB\cap\Omega_j}\bigr].
\end{split}
\]
Since $(-\Delta_{\delta',\beta,\rm  N} -\lambda_0)^{-1}$ is bounded from $L^2(\dR^d)$ into $\dom\fra_{\delta',\beta,\rm N}$ it
follows from \cite[Theorem 3.37]{McLean} that the operator $T_2$ maps $L^2(\dR^d)$ boundedly into 
\[
H^{1/2}(\partial \cB\cap\Omega_1)\oplus H^{1/2}(\partial \cB\cap \Omega_2),
\]
which is compactly embedded in  $L^2(\partial \cB\cap\Omega_1)\oplus L^2(\partial \cB\cap\Omega_2)=L^2(\partial \cB)$. Hence 
$T_2\colon L^2(\dR^d)\rightarrow L^2(\partial \cB)$ is compact. We shall show below in Step II 
that the operator $T_1\colon L^2(\dR^d)\rightarrow L^2(\partial \cB)$
is bounded, so that by \eqref{Wfg5}
$$T_2^*T_1= W = (-\Delta_{\delta',\beta} - \lambda_0)^{-1} - (-\Delta_{\delta',\beta,\rm N} -\lambda_0)^{-1}$$
is a compact operator in $L^2(\dR^d)$. It then follows that the resolvent difference in \eqref{W3} is compact for 
all $\lambda\in\rho(-\Delta_{\delta',\beta})\cap\rho(-\Delta_{\delta',\beta,\rm N})$, see, e.g. \cite[Lemma 2.2]{BLL12a}.

\vskip 0.2cm
\noindent {\bf Step II.}
We verify that $T_1\colon L^2(\dR^d)\rightarrow L^2(\partial \cB)$ is bounded, which is essentially a consequence of
\cite[Theorem 4.18~(ii)]{McLean} and the $H^2$-regularity of the functions in $\dom(-\Delta_{\delta^\prime,\beta})$ near
$\partial\cB\cap\Omega_j$; cf. Lemma~\ref{h2bound}~(ii).
More precisely, let $0<s<t<1$ and let $\cB_s$ and $\cB_t$
be bounded domains with smooth boundaries such that
$$\overline \cB_0\subset\cB_s \subset\overline\cB_s\subset\cB \subset\overline\cB\subset\cB_t\subset\overline\cB_t\subset\cB_1.$$
Set $\cR_{j}:=(\cB_t\setminus\overline\cB_s)\cap\Omega_j$ and $\cS_j:=(\cB_1\setminus\overline\cB_0)\cap\Omega_j$, $j=1,2$.
Since $\Gamma=(\cB_1\setminus\overline\cB_0)\cap\Sigma$ is $C^{1,1}$ we conclude for $u=(-\Delta_{\delta',\beta} - \lambda_0)^{-1} f$ 
from \cite[Theorem 4.18~(ii)]{McLean} that
\begin{equation}\label{estiesti}
 \Vert u_j\vert_{\cR_j}\Vert_{H^2(\cR_j)}\leq C_j\bigl(\Vert u_j\vert_{\cS_j}\Vert_{H^1(\cS_j)} 
 +\Vert \partial_{\nu_j}u_j\vert_\Gamma\Vert_{H^{1/2}(\Gamma)}+\Vert f_j\vert_{\cS_j}\Vert_{L^2(\cS_j)}
 \bigr)
\end{equation}
holds for some constants $C_j$, $j=1,2$. The continuity of $(-\Delta_{\delta',\beta} - \lambda_0)^{-1}$ from $L^2(\dR^d)$
into $\dom\fra_{\delta',\beta}$ yields $\Vert u_j\vert_{\cS_j}\Vert_{H^1(\cS_j)}\leq C^\prime \Vert f\Vert_{L^2(\dR^d)}$ with some constant $C^\prime$.
Furthermore, the boundary conditions $\partial_{\nu_1}u_{1}|_{\Gamma} = \beta^{-1}(u_1|_{\Gamma} - u_2|_{\Gamma})=-\partial_{\nu_2}u_{2}|_{\Gamma}$, 
the continuity of the trace and of $(-\Delta_{\delta',\beta} - \lambda_0)^{-1}$ from $L^2(\dR^d)$
into $\dom\fra_{\delta',\beta}$ yields
\begin{equation*}
 \begin{split}
  \Vert \partial_{\nu_j}u_j\vert_\Gamma\Vert_{H^{1/2}(\Gamma)}& \leq D \bigl(\Vert u_1\vert_\Gamma\Vert_{H^{1/2}(\Gamma)} + 
\Vert u_2\vert_\Gamma\Vert_{H^{1/2}(\Gamma)}\bigr)\\
&\leq D^\prime \bigl(\Vert u_1\vert_{\cS_1}\Vert_{H^1(\cS_1)} +\Vert u_2\vert_{\cS_2}\Vert_{H^1(\cS_2)}\bigr)
\leq D^{\prime\prime} \Vert f\Vert_{L^2(\dR^d)}
 \end{split}
\end{equation*}
with some constants $D,D^\prime,D^{\prime\prime}$.
If $P_j$ denotes the orthogonal projection in $L^2(\dR^d)$ onto $L^2(\cR_j)$ then 
we conclude together with \eqref{estiesti} that 
$$\ran\bigl(P_j (-\Delta_{\delta',\beta} - \lambda_0)^{-1}\bigr)\subset H^2(\cR_j)$$ and that 
the operator $P_j (-\Delta_{\delta',\beta} - \lambda_0)^{-1}$ is bounded from
$L^2(\dR^d)$ into $H^2(\cR_j)$ for $j=1,2$. 
Hence $f\mapsto \partial_{\nu_{j1}}[(-\Delta_{\delta',\beta} - \lambda_0)^{-1}f]_{j1}\vert_{\partial\cB \cap\Omega_j}$
is bounded from $L^2(\dR^d)$ into $H^{1/2}(\partial\cB \cap\Omega_j)$, $j=1,2$, and, in particular, $T_1$ is bounded from 
$L^2(\dR^d)$ into $L^2(\partial\cB)$.

\vskip 0.2cm
\noindent {\bf Step III.}
As in Step I we introduce an auxiliary sesquilinear form by
\begin{equation*}
\begin{split}
\fra^\prime_{\delta',\beta^\prime,\rm N}[h,k]\! &:= \sum_{k,l=1}^{2} (\nabla h_{kl}, \nabla k_{kl})_{L^2(\Omega^\prime_{kl};\dC^d)} \\
& \qquad -\sum_{j=1}^2\big(\beta_j^{\prime\,-1}(h_{1j}|_{\Sigma^\prime_j} -h_{2j}|_{\Sigma^\prime_j}),
k_{1j}|_{\Sigma^\prime_j} -k_{2j}|_{\Sigma^\prime_j}\bigr)_{L^2(\Sigma^\prime_j)},\\
\dom\fra^\prime_{\delta',\beta^\prime,\rm N} &= \bigoplus_{k,l=1}^{2} H^1(\Omega^\prime_{kl}),
\end{split}
\end{equation*}
where $\Omega^\prime_{i1}:=\Omega^\prime_i\cap\cB$, $\Omega^\prime_{i2}:=\Omega^\prime_i\cap(\dR^d\setminus\overline\cB)$, $i=1,2$,
$g_{ij}, h_{ij}$ denote the corresponding restrictions of functions $g,h$, and
$\Sigma^\prime_1:=\Sigma^\prime\cap\cB$, $\Sigma^\prime_2:=\Sigma^\prime\cap(\dR^d\setminus\overline\cB)=\Sigma_2$.
The form $\fra_{\delta',\beta^\prime,\rm N}$ 
is closed, densely defined and semibounded from below,  and hence gives rise to a self-adjoint operator 
$-\Delta^\prime_{\delta',\beta^\prime,\rm N}$ in $L^2(\dR^d)$.
In the same way as in Step I and II one verifies that
\begin{equation}\label{compact3}
(-\Delta^\prime_{\delta',\beta^\prime} -\lambda)^{-1} - (-\Delta^\prime_{\delta',\beta^\prime,\rm N} -\lambda)^{-1}
\end{equation}
is compact for all $\lambda\in(-\Delta^\prime_{\delta',\beta^\prime})\cap(-\Delta^\prime_{\delta',\beta^\prime,\rm N})$.

\vskip 0.2cm
\noindent {\bf Step IV.}
Since the Lipschitz partitions $\cP$ and $\cP^\prime$ are local deformations of each other and Hypothesis~\ref{hypolocal}
holds the self-adjoint operators $-\Delta_{\delta',\beta,\rm N}$ and $-\Delta^\prime_{\delta',\beta^\prime,\rm N}$
from Steps I-III admit the direct sum decompositions
\[
-\Delta_{\delta',\beta,\rm N} = H_1\oplus H_2\quad\text{and}\quad -\Delta^\prime_{\delta',\beta^\prime,\rm N} = H_1\oplus  H_2^\prime
\]
with respect to the decomposition $L^2(\dR^d) = L^2(\dR^d\setminus \ov{\cB})\oplus L^2(\cB)$.
The operators $H_2$ and $H_2^\prime$ acting in $L^2(\cB)$ have
compact resolvents in view of the compact embeddings of the spaces $H^1(\cB\cap\Omega_1)\oplus H^1(\cB\cap\Omega_2)$ and
$H^1(\cB\cap\Omega^\prime_1)\oplus H^1(\cB\cap\Omega^\prime_2)$ into $L^2(\cB)$. This implies the compactness of
\begin{equation*}
(-\Delta_{\delta',\beta,\rm N} -\lambda)^{-1} - (-\Delta^\prime_{\delta',\beta^\prime,\rm N}-\lambda)^{-1}
\end{equation*}
for all $\lambda\in\rho(-\Delta_{\delta',\beta,\rm N})\cap\rho(-\Delta^\prime_{\delta',\beta^\prime,\rm N})$
and hence assertion (ii) follows together with the compactness of the resolvent differences in \eqref{W3} and \eqref{compact3}.
\end{proof}

The following corollary is an immediate consequence of Theorem~\ref{thm:compact} and the fact that for the Lipschitz partition $\cP^\prime=\{\dR^d_+,\dR^d_-\}$ 
and constants $\alpha,\beta>0$ the essential spectra of $-\Delta^\prime_{\delta,\alpha}$ and $-\Delta^\prime_{\delta',\beta}$ can be computed
by separation of variables:
\begin{equation*}
 \begin{split}
  \sigma(-\Delta^\prime_{\delta,\alpha}) &= \sess(-\Delta^\prime_{\delta,\alpha})= [-\alpha^2/4,\infty),\\
  \sigma(-\Delta^\prime_{\delta^\prime,\beta}) &= \sess(-\Delta^\prime_{\delta^\prime,\beta})= [-4/\beta^2 ,\infty).
 \end{split}
\end{equation*}

\begin{cor}\label{cor123}
Let $\cP = \{\Omega_k\}_{k=1}^n$ be a local deformation of the Lipschitz partition $\cP^\prime=\{\dR^d_+,\dR^d_-\}$ of $\dR^d$ 
and let $\alpha,\beta > 0$ be constant.
Then the essential spectra of $-\Delta_{\delta,\alpha}$ and $-\Delta_{\delta',\beta}$ are given by
\[
\sigma_{\rm ess}(-\Delta_{\delta,\alpha}) = [-\alpha^2/4,\infty)\quad\text{and}\quad
\sigma_{\rm ess}(-\Delta_{\delta',\beta}) = [-4/\beta^2 ,\infty).
\]
\end{cor}

The next corollary is a consequence of Theorem~\ref{thm:variation}, Corollary~\ref{thm:main2} and Corollary~\ref{cor123}.

\begin{cor}\label{cor123123}
Let $\cP = \{\Omega_k\}_{k=1}^n$ be a local deformation of the Lipschitz partition $\cP^\prime=\{\dR^d_+,\dR^d_-\}$, 
assume that the chromatic number of $\cP$ is $\chi=2$ and that the constants $\alpha,\beta > 0$ satisfy 
$$\beta = \frac{4}{\alpha},\quad\text{and hence}\quad\sigma_{\rm ess}(-\Delta_{\delta,\alpha}) = \sigma_{\rm ess}(-\Delta_{\delta',\beta})=[-\alpha^2/4,\infty).$$
Let $\{\lambda_{k}(-\Delta_{\delta,\alpha})\}_{k=1}^{\infty}$ and
$\{\lambda_{k}(-\Delta_{\delta',\beta})\}_{k=1}^{\infty}$ be the 
eigenvalues of $-\Delta_{\delta,\alpha}$ and $-\Delta_{\delta',\beta}$ below $-\alpha^2/4$,
respectively, 
and let $N(-\Delta_{\delta,\alpha})$ and $N(-\Delta_{\delta',\beta})$ be their total multiplicities as in Definition~\ref{dfn:spec1}.
Then the following statements hold:
\begin{itemize}\setlength{\itemsep}{1.2ex}
\item [\rm (i)] $\lambda_k(-\Delta_{\delta',\beta})\le \lambda_k(-\Delta_{\delta,\alpha})$ for all $k\in\dN$;
\item[\rm (ii)] $N(-\Delta_{\delta,\alpha})\leq N(-\Delta_{\delta',\beta})$.
\end{itemize}
\end{cor}

\subsection{Locally deformed partitions of $\dR^2$ and $\dR^3$}\label{sec23case}
In this subsection special attention is paid to bound states of $-\Delta_{\delta,\alpha}$ and $-\Delta_{\delta',\beta}$
induced by local deformations of certain Lipschitz partitions of $\dR^2$ and $\dR^3$. We first
characterize the essential spectra of $-\Delta_{\delta,\alpha}$ and 
$-\Delta_{\delta',\beta}$ 
associated with partitions, which are local deformations of a partition $\{\Omega,\dR^2\setminus\ov\Omega\}$ with $\Omega$ being a wedge, see Figure~\ref{fig5}. 
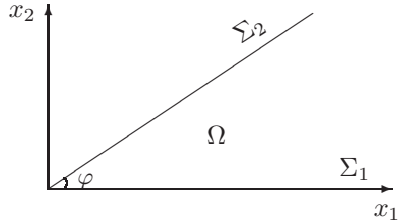
\begin{figure}[H]
\begin{center}
\begin{picture}(200,80)
\put(20,10){\vector(1,0){130}}
\put(20,10){\vector(0,1){70}}
\put(20,10){\line(3,2){100}}
\put(80,27){$\Omega$}
\put(89,65){\begin{turn}{33.3}$\Sigma_2$\end{turn}}
\put(130,15){\begin{turn}{0}$\Sigma_1$\end{turn}}
\put(143,0){\begin{turn}{0}$x_1$\end{turn}}
\put(5,75){\begin{turn}{0}$x_2$\end{turn}}
\qbezier(26,14)(28,13)(26,10)
\put(31,12){\small{$ \varphi$}}
\end{picture}
\end{center}
\caption{A wedge $\Omega\subset\dR^2$ with angle $\varphi\in(0,\pi]$ and boundary consisting of the two rays $\Sigma_1$ and $\Sigma_2$;
the axis $x_1$ coincides with the ray $\Sigma_1$.}
\label{fig5}
\end{figure}
\begin{thm}\label{thm:essspec}
Let $\cP = \{\Omega_k\}_{k=1}^n$ be a local deformation of the Lipschitz partition $\cP' = \{\Omega,\dR^2\setminus\ov{\Omega}\}$ of $\dR^2$, 
where $\Omega$ is a wedge in $\dR^2$ and let $\alpha,\beta > 0$ be constant.
Then the essential spectra of $-\Delta_{\delta,\alpha}$ and $-\Delta_{\delta',\beta}$ are given by
\[
\sigma_{\rm ess}(-\Delta_{\delta,\alpha}) = [-\alpha^2/4,\infty)\quad\text{and}\quad
\sigma_{\rm ess}(-\Delta_{\delta',\beta}) = [-4/\beta^2 ,\infty).
\]
\end{thm}

\begin{proof}
According to Theorem~\ref{thm:compact} it suffices to show the statements for the operators $-\Delta^\prime_{\delta,\alpha}$ 
and $-\Delta^\prime_{\delta',\beta}$ associated with
the Lipschitz partition $\cP' = \{\Omega,\dR^2\setminus\ov{\Omega}\}$. In fact, the assertion for 
$-\Delta^\prime_{\delta,\alpha}$ can be found in \cite[Proposition 5.4]{EN03}, and hence we verify 
$\sigma_{\rm ess}(-\Delta^\prime_{\delta',\beta}) = [-4/\beta^2 ,\infty)$ only.

\noindent {\bf Step I.}
Decompose $\dR^2$ into eight domains as in Figure~\ref{fig4}, 
where $\Omega_1$, $\Omega_1'$, $\Omega_2$, $\Omega_2'$ coincide (up to rotations and translations) with $[0,l]\times \dR_+$ for some $l>0$;
 $\Omega_3$, $\Omega_3'$ are bounded Lipschitz domains, and $\Omega_4$ and $\Omega_5$ are wedges with angles $\varphi$ and $2\pi - \varphi$, respectively.
We choose this partition in such a way that $\Omega$ coincides with $\Omega_1\cup\Omega_2\cup\Omega_3\cup\Omega_4$ up to a set of 
Lebesgue measure zero.
\begin{figure}[H]
\begin{center}
\begin{picture}(200,160)
\thinlines
\multiput(100,70)(-8,16){5}{\line(-1,2){6}}
\multiput(100,70)(8,16){5}{\line(1,2){6}}
\multiput(100,10)(-8,16){8}{\line(-1,2){6}}
\multiput(100,10)(8,16){8}{\line(1,2){6}}
\multiput(97,68)(-5,-3){5}{\begin{turn}{32.3}\line(1,0){3}\end{turn}}
\multiput(99,70)(5,-3){5}{\begin{turn}{-32.3}\line(1,0){3}\end{turn}}
\thicklines
\put(100,40){\line(-1,2){49}}
\put(100,40){\line(1,2){49}}
\put(95,100){$\Omega_4$}
\put(15,50){$\Omega_5$}
\put(60,90){$\Omega_1'$}
\put(72,99){$\Omega_1$}
\put(118,100){$\Omega_2$}
\put(128,90){$\Omega_2'$}
\put(95,56){$\Omega_3$}
\put(95,27){$\Omega_3'$}
\end{picture}
\end{center}
\caption{A partition of $\dR^2$ into eight domains. The wedge $\Omega$ coincides with $\Omega_1\cup\Omega_2\cup\Omega_3\cup\Omega_4$ up 
to a set of Lebesgue measure zero.}
\label{fig4} 
\end{figure}
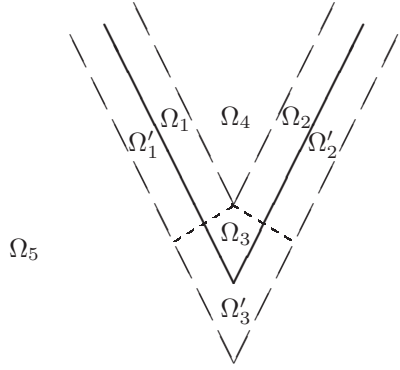
Let $\cP$ be the corresponding partition and set $\Sigma_k := \ov{\Omega_k}\cap\ov{\Omega_k'}$ for $k=1,2,3$. Clearly,
$\partial\Omega=\Sigma_1\cup\Sigma_2\cup\Sigma_3$.
Observe that such a decomposition can be constructed for any $l >0$. We use the notation $f_\Omega := f|_{\Omega}$. Consider the quadratic form
\[
\begin{split}
\fra_{\delta',\beta, \rm N}'[f] &:=\sum_{\Omega\in\cP} \|\nabla f_\Omega\|^2_{L^2(\Omega;\dC^d)} -\sum_{k=1}^3\beta^{-1}\|f_{\Omega_k}|_{\Sigma_k}- f_{\Omega_k'}|_{\Sigma_{k}}\|^2_{L^2(\Sigma_{k})},\\
\dom\fra_{\delta',\beta, \rm N}' &:= \oplus_{\Omega\in\cP}H^1(\Omega)
\end{split}
\]
Similarly as in the proof of Proposition~\ref{prop:deltaform} one verifies that the form
$\fra_{\delta',\beta, \rm N}'$ is closed, densely defined, symmetric and semibounded from below. 
The corresponding self-adjoint operator $-\Delta_{\delta',\beta, \rm N}'$ can be decomposed into the orthogonal sum
of five self-adjoint operators
\begin{equation}
\label{orth}
-\Delta_{\delta',\beta, \rm N}' = \bigoplus_{k=1}^5 H_k,
\end{equation}
where $H_i$ acts in $L^2(\Omega_i)\oplus L^2(\Omega_i')$, $i=1,2,3$, and $H_4$ and $H_5$ are the 
self-adjoint Neumann Laplacians on the wedges $\Omega_4$ and $\Omega_5$
in $L^2(\Omega_4)$ and $L^2(\Omega_5)$, respectively. 
Hence we have
\begin{equation}
\label{H12}
\sess(H_4) = \sess(H_5) = [0,+\infty). 
\end{equation}
The operator $H_3$ acts on a bounded domain and in view of the compact embedding
of the space $H^1(\Omega_3)\oplus H^1(\Omega_3')$ into $L^2(\Omega_3)\oplus L^2(\Omega_3')$ we obtain
\begin{equation}
\label{H5}
\sess(H_3) = \varnothing.
\end{equation}
Separation of variables shows that
the essential spectra of the operators $H_1$ and $H_2$ have the form
\begin{equation}
\label{H34}
\sess(H_1) = \sess(H_2) = [\varepsilon(\beta,l),+\infty),
\end{equation}
where $\varepsilon(\beta,l)$ is the principal eigenvalue of the self-adjoint one-dimensional Schr\"o\-din\-ger operator on the interval $(-l,l)$ with
Neumann boundary conditions at the endpoints $-l$ and $l$, and a $\delta'$-interaction of strength $\beta$ at the origin. According to \cite[Lemma 3.3]{EJ13}
\begin{equation}
\label{EJ}
\varepsilon(\beta,l) < -\frac{4}{\beta^2}\quad\text{and}\quad \lim\limits_{l\rightarrow +\infty} \varepsilon(\beta,l) = -\frac{4}{\beta^2}.
\end{equation}
From the decomposition \eqref{orth} and the characterizations \eqref{H12}, \eqref{H5}, \eqref{H34} we conclude
\[
\sess(-\Delta_{\delta',\beta, \rm N}') = [\varepsilon(\beta,l),+\infty).
\]
Clearly, $
\fra_{\delta',\beta, \rm N}' \le \fra_{\delta',\beta}'$
holds in the sense of Definition~\ref{dfn:forms} and hence 
\[
\min\sess(-\Delta_{\delta',\beta}')\ge \varepsilon(\beta,l)
\]
by Theorem~\ref{thm:variation}\,(ii).
As we noted above, the construction in the proof can be realized for any $l >0$. Thus by \eqref{EJ}
\[
\min\sess(-\Delta_{\delta',\beta}')\ge -\frac{4}{\beta^2}.
\]
{\bf Step II.} In view of Step I it suffices to show that
for any $\lambda\in [-4/\beta^2,+\infty)$ there exists a  singular sequence for the operator $-\Delta_{\delta',\beta}'$ corresponding to $\lambda$. Let us fix the axes $(x_1,x_2)$ such that the axis $x_1$ coincides with the side $\Sigma_1$ of the wedge $\Omega$, see Figure~\ref{fig5}.
Let us fix two functions $\varphi_1, \varphi_2 \in C^\infty_0([0,\infty))$ with $\supp \varphi_1$ and $\supp \varphi_1$ in $[0,2)$ such that $\varphi_1(x) = \varphi_2(x) =1$ in the vicinity of the point $x = 0$
and $0\le \varphi_2(x) \le 1$. Consider the sequence of functions
\[
\psi_{n,p}(x) := \frac{1}{\sqrt{n}}\varphi_1\Big(\frac{1}{n}|x_1-x_{1}^{(n)}|\Big)\varphi_2\Big(\frac{1}{n}|x_2|\Big)
\sign(x_2)e^{-2|x_2|/\beta} e^{ipx_1},\quad n\in\dN,
\]
where $p \ge 0$ is arbitrary and the sequence  $\{x^{(n)}_1\}$ tends to $+\infty$ sufficiently fast, so that the sequence of the supports 
$\supp \psi_{n,p}$ does not intersect the ray $\Sigma_2$ of the wedge. We denote
by $\psi_{n,p,\Omega}$ and $\psi_{n,p,\dR^2\setminus\ov\Omega}$ the restriction of $\psi_{n,p}$ onto $\Omega$ and $\dR^2\setminus\ov\Omega$,
respectively. Computing the traces of $\psi_{n,p}$ from both sides of $\Sigma_1$ we find
\[
\partial_\nu\psi_{n,p, \Omega}|_{\Sigma_1} = \frac{2}{\beta}\frac{1}{\sqrt{n}}\varphi_1\Big(\frac{1}{n}|x_1-x_{1}^{(n)}|\Big)e^{ipx_1} =
\frac{2}{\beta}\psi_{n,p, \Omega}|_{\Sigma_1} =  -\frac{2}{\beta}\psi_{n,p, \dR^2\setminus\ov\Omega}|_{\Sigma_1}
\]
with the normal $\nu$ pointing outwards of $\Omega$.
Thus we conclude from Theorem~\ref{thm:delta}\,(ii) that the functions $\psi_{n,p}$ are in $\dom(-\Delta_{\delta',\beta}')$. Obviously, the sequence of the functions $\{\psi_{n,p}\}$ converges weakly to zero. 
Moreover, with the help of the dominated convergence theorem we get 
\[
\lim\limits_{n\rightarrow\infty}\|\psi_{n,p}\|^2 = \|\varphi_1\|_{L^2(\dR)}^2\int_{\dR}e^{-4|x|/\beta}dx = \frac{\beta}{2}\|\varphi_1\|_{L^2(\dR)}^2\ne 0.
\]
One can check via direct computation that
\[
-\Delta_{\delta',\beta}' \psi_{n,p} = \Big(-\frac{4}{\beta^2} + p^2\Big)\psi_{n,p} + O\Big(\frac{1}{n}\Big),\quad n\rightarrow\infty,
\]
which yields
\begin{equation}
\label{limit}
\|(-\Delta_{\delta',\beta}' + 4/\beta^2 - p^2)\psi_{n,p}\|_{L^2(\dR^2)}\rightarrow 0,\quad n\rightarrow +\infty.
\end{equation}
Therefore, the sequence
\[
\wt\psi_{n,p} := \frac{\psi_{n,p}}{\|\psi_{n,p}\|},\quad n\in\dN,
\]
is a singular sequence for the operator $-\Delta_{\delta',\beta}'$
corresponding to the point $-4/\beta^2 + p^2$.
Since the choice of $p$ is arbitrary, the claim is proven.
\end{proof}

The next corollary is a consequence of Theorems~\ref{thm:variation} and~\ref{thm:essspec}.

\begin{cor}
\label{cor1234}
Let $\cP = \{\Omega_k\}_{k=1}^n$ be a local deformation of the Lipschitz partition $\cP' = \{\Omega,\dR^2\setminus\ov{\Omega}\}$ 
with $\Omega$ being a wedge, assume that the chromatic number of $\cP$ is $\chi=2$ and that the constants $\alpha,\beta > 0$ satisfy 
$$\beta = \frac{4}{\alpha},\quad\text{and hence}\quad\sigma_{\rm ess}(-\Delta_{\delta,\alpha}) = \sigma_{\rm ess}(-\Delta_{\delta',\beta})=[-\alpha^2/4,\infty).$$
Let $\{\lambda_{k}(-\Delta_{\delta,\alpha})\}_{k=1}^{\infty}$ and
$\{\lambda_{k}(-\Delta_{\delta',\beta})\}_{k=1}^{\infty}$ be the eigenvalues of $-\Delta_{\delta,\alpha}$ and $-\Delta_{\delta',\beta}$
below $-\alpha^2/4$,
respectively, 
and let $N(-\Delta_{\delta,\alpha})$ and $N(-\Delta_{\delta',\beta})$ be their total multiplicities as in Definition~\ref{dfn:spec1}.
Then the following statements hold:
\begin{itemize}\setlength{\itemsep}{1.2ex}
\item [\rm (i)] $\lambda_k(-\Delta_{\delta',\beta})\le \lambda_k(-\Delta_{\delta,\alpha})$ for all $k\in\dN$;
\item[\rm (ii)] $N(-\Delta_{\delta,\alpha})\leq N(-\Delta_{\delta',\beta})$.
\end{itemize}
\end{cor}

The following corollary shows the existence of negative bound states of $-\Delta_{\delta',\beta}$ for locally deformed broken lines in $\dR^2$.
The assertion follows directly from \cite[Theorem 5.2]{EI01} and Corollary~\ref{cor1234}. We mention that 
in \cite{EI01} more general weakly deformed curves were considered.

\begin{cor}
Let $\cP = \{\Omega,\dR^2\setminus\ov\Omega\}$  be a local deformation of the Lipschitz partition $\cP' = \{\Omega',\dR^2\setminus\ov{\Omega'}\}$, 
where $\Omega'$ is a wedge with angle $\varphi\in(0,\pi]$. In the case $\varphi = \pi$ let $\cP \ne \cP'$.
Assume, in addition, that $\partial\Omega$ is piecewise $C^1$-smooth. Then $N(-\Delta_{\delta',\beta}) \ge 1$ holds for any $\beta > 0$.
\end{cor}

In the next proposition we show the existence of bound states for $\delta$ and $\delta'$-operators
for special local deformations of the partition $\{\dR^2_+,\dR^2_-\}$.

\begin{prop}
Let $\Omega_1\subset\dR^2_+$ be a bounded Lipschitz domain and consider 
the Lipschitz partition $\cP = \{\Omega_k\}_{k=1}^3$ of $\dR^2$, where $\Omega_2=\dR^2_+\setminus\overline\Omega_1$ 
and $\Omega_3=\dR^2_-$ as in Figure~\ref{hyphyp}.
\noindent 
Let $\alpha,\beta > 0$ be constant and let the Schr\"odinger operators
$-\Delta_{\delta,\alpha}$ and $-\Delta_{\delta',\beta}$ be associated with
$\cP$. Then the following statements hold.
\begin{itemize}\setlength{\itemsep}{1.2ex}
\item [\rm (i)] $\sess(-\Delta_{\delta,\alpha}) = [-\alpha^2/4,+\infty)$ and  $N(-\Delta_{\delta,\alpha})\ge 1$;
\item [\rm (ii)] $\sess(-\Delta_{\delta',\beta}) = [-4/\beta^2,+\infty)$ and $N(-\Delta_{\delta',\beta})\ge 1$.
\end{itemize}
\end{prop}
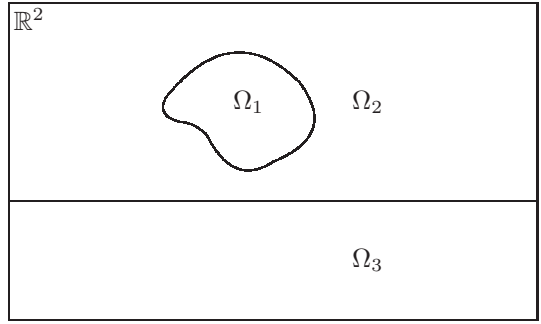
\begin{figure}[H]
\begin{center}
\begin{picture}(200,120)
\qbezier(60,85)(85,115)(110,90)
\qbezier(110,90)(125,70)(100,60)
\qbezier(100,60)(85,50)(75,70)
\qbezier(75,70)(70,75)(65,75)
\qbezier(60,85)(55,77.5)(65,75)
\put(85,80){$\Omega_{\rm 1}$}
\put(130,80){$\Omega_{\rm 2}$}
\qbezier(0,45)(100,45)(200,45)
\put(130,20){$\Omega_3$}
\qbezier(0,0)(0,60)(0,120)
\qbezier(200,0)(200,60)(200,120)
\qbezier(0,0)(100,00)(200,0)
\qbezier(0,120)(100,120)(200,120)
\put(2,110){$\dR^2$}
\end{picture}
\end{center}
\caption{The partition of $\dR^2$ via a straight line and a compact Lipschitz contour, which consists of a bounded domain $\Omega_1$ and 
two unbounded domains $\Omega_2$ and $\Omega_3$.}
\label{hyphyp}
\end{figure}

\begin{proof}
The characterization of the essential spectra in (i) and (ii) is a direct consequence of Corollary~\ref{cor123}. 
We shall show the assertion $N(-\Delta_{\delta,\alpha})\ge 1$ in (i) first. For this we can assume that
$\Sigma_{23}$ is the hyperplane defined by $x_2 = 0$, where $x=(x_1,x_2)\in\dR^2$.
Let $\varphi\in C^\infty_0(\dR)$ be equal to one in the neighbourhood of the origin and consider the sequence of functions
\[
f_n(x) := \varphi\Big(\frac{1}{n}x_1\Big)e^{-(\alpha/2)|x_2|}\in H^1(\dR^2),\quad n\in\dN,
\]
and the sequence of real values
$$I_n := \fra_{\delta,\alpha}[f_n] + \frac{\alpha^2}{4}\|f_n\|^2_{L^2(\dR^2)},\quad n\in\dN.$$
From the definition of the form $\fra_{\delta,\alpha}$ in \eqref{deltaform},
\begin{equation*}
 \begin{split}
  \|f_n\|^2_{L^2(\dR^2)}&=\frac{2n}{\alpha}\|\varphi\|_{L^2(\dR)}^2,\\
  \|\nabla f_n\|^2_{L^2(\dR^2;\dC^2)}&=\frac{2}{n\alpha}\|\varphi^\prime\|_{L^2(\dR)}^2+\frac{\alpha n}{2}\|\varphi\|_{L^2(\dR)}^2,
 \end{split}
\end{equation*}
and $\|f_n|_{\Sigma_{23}}\|_{L^2(\Sigma_{23})}^2=n\|\varphi\|_{L^2(\dR)}^2$ we obtain 
\begin{equation*}
\begin{split}
I_n &= \|\nabla f_n\|^2_{L^2(\dR^2;\dC^2)}- \alpha\|f_n|_{\Sigma_{12}}\|_{L^2(\Sigma_{12})}^2 -\alpha \|f_n|_{\Sigma_{23}}\|_{L^2(\Sigma_{23})}^2 + \frac{\alpha^2}{4}\|f_n\|^2_{L^2(\dR^2)}\\ 
&=\frac{2}{\alpha n}\|\varphi'\|_{L^2(\dR)}^2 - \alpha \|f_n|_{\Sigma_{12}}\|^2_{L^2(\Sigma_{12})}.
\end{split}
\end{equation*}
Let $D > 0$ be the distance from $\Sigma_{23}$ to the farthest point of $\Sigma_{12}$.  For large $n\in\dN$
\[
I_n  \le \frac{2}{\alpha n}\|\varphi'\|_{L^2(\dR)}^2 - \alpha e^{-\alpha D}|\Sigma_{12}|,
\]
and hence for sufficiently large $n\in\dN$
we obtain that $I_n < 0$, which proves the existence of at least one bound state for the operator $-\Delta_{\delta,\alpha}$.

In order to show that $-\Delta_{\delta',\beta}$ has at least one bound state we note that $-\Delta_{\delta,4\beta^{-1}}$ has at least one bound state
by the considerations above. Hence Corollary~\ref{cor1234} implies 
$N(-\Delta_{\delta',\beta})\ge 1$ . 
\end{proof}

The next result shows the existence of negative bound states of $-\Delta_{\delta',\beta}$ for certain 
hypersurfaces in $\dR^3$. The assertion follows directly from  and \cite[Theorem 4.3]{EK03} and Corollary~\ref{cor123123}. 
We mention that in \cite{EK03} more general weakly deformed planes were considered.

\begin{cor}
Let $\cP = \{\Omega,\dR^3\setminus\ov\Omega\}$ be a local deformation of the Lipschitz partition $\cP' = \{\dR^3_+,\dR^3_-\}$ such 
that $\cP \ne \cP'$. If, in addition, $\partial\Omega$ is $C^4$-smooth and admits a global natural parametrization 
in the sense of \cite[\S 2-3, Definition 2]{doCarmo} then $N(-\Delta_{\delta',\beta}) \ge 1$ for all sufficiently small $\beta > 0$. 
\end{cor}

\section{Appendix: Sobolev spaces on wedges and a symmetric star graph with three leads in $\dR^2$}

In this appendix we verify the statements
\begin{equation*}
\min(-\Delta_{\delta,\alpha})=-\frac{\alpha^2}{3}\quad\text{and}\quad 
\min(-\Delta_{\delta^\prime,\beta}) \geq - \Bigg(\frac{12\sqrt{3}-2}{9}\Bigg)^2\frac{1}{\beta^2}
\end{equation*}
from \eqref{mindelta} and \eqref{mindeltaprime}, where $\alpha,\beta>0$ are real constants and the $\delta$ and $\delta^\prime$-interaction
is supported on the symmetric star graph with three leads in Figure~\ref{stargraph}. We first provide some useful estimates and decompositions
for $H^1$-functions on wedges.

\subsection{Sobolev spaces on wedges in $\dR^2$}

\label{ssec:Sobolev2}
Let $\Omega$ be a wedge with angle $\varphi\in(0,\pi]$ as in the figure below. The estimates for functions $f\in H^1(\Omega)$
in Lemma~\ref{lem:wedgetrace1} and Lemma~\ref{lem:wedgetrace2} below will be used in the proofs of
\eqref{mindelta} and \eqref{mindeltaprime}.

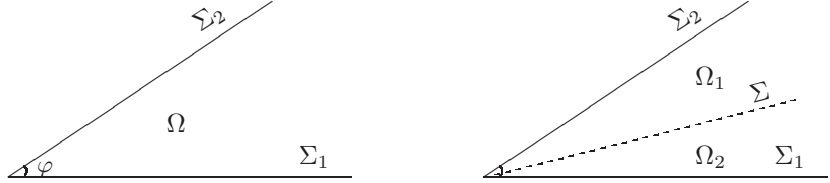
\begin{figure}[H]
\begin{center}
\begin{picture}(400,80)
\put(20,10){\line(1,0){130}}
\put(20,10){\line(3,2){100}}
\put(80,27){$\Omega$}
\put(89,65){\begin{turn}{33.3}$\Sigma_2$\end{turn}}
\put(130,15){\begin{turn}{0}$\Sigma_1$\end{turn}}
\qbezier(26,14)(28,13)(26,10)
\put(31,12){\small{$ \varphi$}}
\put(200,10){\line(1,0){130}}
\put(200,10){\line(3,2){100}}
\put(280,45){$\Omega_1$}
\put(280,15){$\Omega_2$}
\put(269,65){\begin{turn}{33.3}$\Sigma_2$\end{turn}}
\put(310,15){\begin{turn}{0}$\Sigma_1$\end{turn}}
\put(300,38){\begin{turn}{15.0}$\Sigma$\end{turn}}
 \multiput(200,10)(4,1){30}{\begin{turn}{15.0}\line(1,0){2}\end{turn}}
\qbezier(206,14)(208,13)(206,10)
\end{picture}
\end{center}
\caption{A wedge $\Omega\subset\dR^2$ with angle $\varphi\in(0,\pi]$ and boundary consisting of the two rays $\Sigma_1$ and $\Sigma_2$;
the ray $\Sigma$ separates $\Omega$ into two wedges $\Omega_1$ and $\Omega_2$.}
\label{fig:wedge}
\end{figure}

The first lemma is a reformulation of \cite[Lemma 2.6]{LP08}.

\begin{lem}\label{lem:wedgetrace1}
Let $\Omega\subset\dR^2$
be a wedge with angle $\varphi\in(0,\pi]$ and boundary $\partial\Omega$. Then for every $f\in H^1(\Omega)$ 
and all $\gamma >0$ the estimate
\[
\|\nabla f\|^2_{L^2(\Omega;\dC^2)}- \gamma\|f|_{\partial\Omega}\|^2_{L^2(\partial\Omega)} \ge -\frac{\gamma^2}{\sin^2(\varphi/2)}\|f\|_{L^2(\Omega)}^2
\]
holds. For $\varphi\in(0,\pi)$ the estimate is sharp.
\end{lem}

We provide a variant of Lemma~\ref{lem:wedgetrace1} which will be useful in the proof of \eqref{mindeltaprime}.
We note for completeness that the estimate below is not sharp for $\varphi\in(0,\pi)$.

\begin{lem}\label{lem:wedgetrace2}
Let $\Omega\subset\dR^2$
be a wedge with angle $\varphi\in(0,\pi]$ and boundary
$\partial\Omega$.
Let $\Sigma$ be a ray separating $\Omega$ into two wedges as in Figure \ref{fig:wedge}. 
Then for every $f\in H^1(\Omega)$ with $f|_{\Sigma} = 0$
and all $\gamma >0$ the estimate
\[
\|\nabla f\|^2_{L^2(\Omega;\dC^2)}-\gamma\|f|_{\partial\Omega}\|^2_{L^2(\partial\Omega)} \ge - \gamma^2\|f\|_{L^2(\Omega)}^2
\]
holds. 
\end{lem}

\begin{proof}
Let $f\in H^1(\Omega)$ with $f|_{\Sigma} = 0$ and denote the restrictions of $f$ to the wedges 
$\Omega_1$ and $\Omega_2$ by $f_1$ and $f_2$,
respectively. Extend the wedge $\Omega_1$ with degree $\varphi_1 <\varphi$ to the half-plane $\dR^2_+$ 
by gluing the wedge $\Omega_1'$ with the
angle $\pi - \varphi_1$ as in Figure~\ref{gluing}.

\begin{figure}[H]
\label{fig1}
\begin{center}
\begin{picture}(250,80)
\put(100,10){\line(1,0){130}}
\put(100,10){\line(3,2){100}}
\multiput(100,10)(-7,0){14}{\line(-1,0){3}}  
\put(160,27){$\Omega_1$}
\put(60,27){$\Omega_1'$}
\put(169,65){\begin{turn}{33.3}$\Sigma$\end{turn}}
\put(210,15){\begin{turn}{0}$\Sigma_1$\end{turn}}
\qbezier(106,14)(108,13)(106,10)
\put(110,12){\small{$ \varphi_1$}}
\qbezier(105,15)(98,15)(97,10)
\put(110,12){\small{$ \varphi_1$}}
\put(75,14){\small{$\pi\!\!-\!\! \varphi_1$}}

\end{picture}
\end{center}
\caption{Extension of the wedge $\Omega_1$ to the half-plane $\dR^2_+$ via the wedge $\Omega_1'$.}
\label{gluing}
\end{figure}
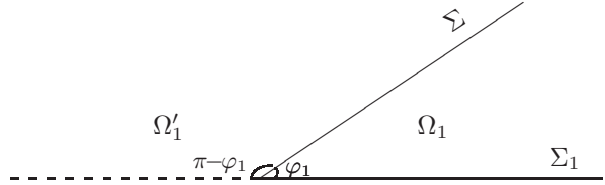
As $f_1\in H^1(\Omega_1)$ and $f_1|_{\Sigma} = 0$ we can extend $f_1$ by zero to  
$\wt f\in H^1(\dR^2_+)$. Then for $\gamma > 0$
\begin{equation*}
\begin{split}
 \|\nabla f_1\|^2_{L^2(\Omega_1;\dC^2)}-\gamma\|f_1|_{\Sigma_1}\|_{L^2(\Sigma_1)}^2&=
 \|\nabla \wt f\|^2_{L^2(\dR^2_+;\dC^2)}-\gamma\|\wt f|_{\partial\dR^2_+}\|_{L^2(\partial\dR^2_+)}^2\\
 & \ge - \gamma^2\|\wt f\|_{L^2(\dR^2_+)}^2=- \gamma^2\|f_1\|_{L^2(\Omega_1)}^2
\end{split}
 \end{equation*} holds
by Lemma~\ref{lem:wedgetrace1}. The same argument shows that for $\gamma > 0$ the function 
$f_2\in H^1(\Omega_2)$ satisfies
\begin{equation*}
 \|\nabla f_2\|^2_{L^2(\Omega_2;\dC^2)}-\gamma\|f_2|_{\Sigma_2}\|_{L^2(\Sigma_1)}^2 \ge - \gamma^2\|f_2\|_{L^2(\Omega_2)}^2.
 \end{equation*}
Summing up the above estimates we obtain the estimate in the lemma.
\end{proof}

It turns out to be useful in the proof of \eqref{mindeltaprime} to decompose functions in $H^1(\Omega)$ 
as sums of even and odd functions with respect to the angle bisector of the wedge $\Omega$.

\begin{lem}\label{lem:decomposition}
Let $\Omega\subset\dR^2$ be a wedge with angle $\varphi \in (0,\pi]$ and let $\Sigma$ be the angle bisector which separates $\Omega$ 
into two wedges with angles $\varphi/2 \in (0,\pi/2]$.
Then every $f \in H^1(\Omega)$ can be decomposed into the sum $f_{\rm o} + f_{\rm e}$ such that the following conditions
{\rm (a)}-{\rm (e)} hold:
\begin{itemize}\setlength{\parskip}{1.2mm}
\item [\rm (a)] $f_{\rm o}, f_{\rm e} \in H^1(\Omega)$;
\item [\rm (b)] $(f_{\rm o}, f_{\rm e})_{L^2(\Omega)} = 0$;
\item [\rm (c)] $(\nabla f_{\rm o},\nabla f_{\rm e})_{L^2(\Omega;\dC^2)} = 0$;
\item [\rm (d)] $f_{\rm e}|_{\Sigma_1} = f_{\rm e}|_{\Sigma_2}$ and
$f_{\rm o}|_{\Sigma_1} = - f_{\rm o}|_{\Sigma_2}$;
\item [\rm (e)] $f_{\rm o}|_{\Sigma} = 0$.
\end{itemize}
\end{lem}

\begin{proof}
Fix the Cartesian coordinate system such that the vertex of the wedge $\Omega$ is the origin and
the angle bisector is the ordinate axis as in Figure~\ref{cartesian}
\begin{figure}[H]
\label{fig}
\begin{center}
\begin{picture}(200,100)
\put(100,100){\line(1,-2){40}}
\put(100,100){\line(-1,-2){40}}
\put(100,100){\vector(0,-1){90}}
\put(100,100){\vector(1,0){70}}
\put(85,50){$\Omega$}
\put(55,30){\begin{turn}{66.6}$\Sigma_1$\end{turn}}
\put(135, 35){\begin{turn}{-60}$\Sigma_2$\end{turn}}
\put(105,12){\small{$y$}}
\put(165,90){\small{$x$}}
\put(92,99){\small{$O$}}
\end{picture}
\end{center}
\caption{The wedge $\Omega$ in the Cartesian coordinate system.}
\label{cartesian}
\end{figure}
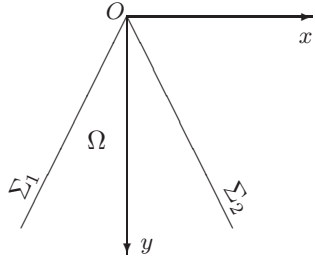
Let
$C^\infty_0(\ov\Omega) := \{f|_{\Omega}\colon f\in C^\infty_0(\dR^2)\}$ and 
note that for $h_1,h_2\in C^\infty_0(\ov\Omega)$ with
\[
h_1(x,y) = h_1(-x,y)\quad\text{and}\quad h_2(x,y) = - h_2(-x,y)
\]
the equality
\begin{equation}
\label{zero}
(h_1,h_2)_{L^2(\Omega)} = 0
\end{equation}
holds. Let us introduce the mappings
\[
\begin{split}
P_{\rm even, 0}\colon C^\infty_0(\ov\Omega)\rightarrow C^\infty_0(\ov\Omega),&\qquad
(P_{\rm even, 0}f)(x,y) := \frac{f(x,y) + f(-x,y)}{2},\\
P_{\rm odd, 0}\colon C^\infty_0(\ov\Omega)\rightarrow C^\infty_0(\ov\Omega),&\qquad
(P_{\rm odd, 0}f)(x,y) := \frac{f(x,y) - f(-x,y)}{2},
\end{split}
\]
and define $f_{\rm e} := P_{\rm even,0}f$ and $f_{\rm o} := P_{\rm odd,0}f$. Then obviously, $f = f_{\rm e} +f_{\rm o}$, and 
$f_{\rm e}$ and $f_{\rm o}$ satisfy (d) and (e) by
their definition. Condition (b) holds according to \eqref{zero}.
Computing partial derivatives we get 
\[
\begin{split}
\partial_1 f_{\rm e} = \frac{(\partial_1 f)(x,y) - (\partial_1 f)(-x,y)}{2},&\qquad
\partial_1 f_{\rm o} = \frac{(\partial_1 f)(x,y) + (\partial_1 f)(-x,y)}{2},\\
\partial_2 f_{\rm e} = \frac{(\partial_2 f)(x,y) + (\partial_2 f)(-x,y)}{2},&\qquad
\partial_2 f_{\rm o} = \frac{(\partial_2 f)(x,y) - (\partial_2 f)(-x,y)}{2}.
\end{split}
\]
Hence, according to \eqref{zero} we conclude 
\[
(\partial_1 f_{\rm e}, \partial_1 f_{\rm o})_{L^2(\Omega)} = 0\quad \text{and}\quad (\partial_2 f_{\rm e}, \partial_2 f_{\rm o})_{L^2(\Omega)} = 0.
\]
Therefore, condition (c) holds as well. Note that
\[
\|P_{\rm even,0}f\|_{H^1(\Omega)} \le \|f\|_{H^1(\Omega)}\quad\text{and}\quad\|P_{\rm odd,0}f\|_{H^1(\Omega)} \le \|f\|_{H^1(\Omega)},
\]
where $\|f\|_{H^1(\Omega)}^2 = \|f\|^2_{L^2(\Omega)} + \|\nabla f\|^2_{L^2(\Omega;\dC^2)}$. Hence, the operators
$P_{\rm even,0}$ and $P_{\rm odd,0}$ can be extended by continuity to operators 
$P_{\rm even}\colon H^1(\Omega)\rightarrow H^1(\Omega)$ and $P_{\rm odd}\colon H^1(\Omega)\rightarrow H^1(\Omega)$ 
with $\dom P_{\rm even} = \dom P_{\rm odd} = H^1(\Omega)$.
For $f\in H^1(\Omega)$ define $f_{\rm e} := P_{\rm even}f$ and $f_{\rm o} := P_{\rm odd}f$. 
Clearly, $f = f_{\rm e} + f_{\rm o}$ holds, and the conditions (a), (b), (c),  (d) and (e) are satisfied. 
This completes the proof of the lemma.
\end{proof}

\subsection{Proof of \eqref{mindelta}}

The assertion $\min\sigma(-\Delta_{\delta,\alpha}) = - \frac{\alpha^2}{3}$ is essentially a consequence of Lemma~\ref{lem:wedgetrace1}.
In fact, recall first that the operator $-\Delta_{\delta,\alpha}$ corresponds to the quadratic form
\[
\fra_{\delta,\alpha}[f] = \|\nabla f\|^2_{L^2(\dR^2;\dC^2)} - \alpha\|f|_{\Sigma}\|^2_{L^2(\Sigma)},\quad \dom\fra_{\delta,\alpha}  = H^1(\dR^2).
\]
Let $f\in\dom\fra_{\delta,\alpha}$ such that $\|f\|_{L^2} = 1$, 
and denote the restrictions by $f_k := f|_{\Omega_k}$, $k=1,2,3$. 
From Lemma~\ref{lem:wedgetrace1} with $\gamma = \alpha/2$ and $\varphi = 2\pi/3$ we obtain the estimates
\begin{equation}\label{123wedge}
\begin{split}
\|\nabla f_k\|^2_{L^2(\Omega_k;\dC^2)} - \frac{\alpha}{2}\|f_k|_{\partial\Omega_k}\|^2_{L^2(\partial\Omega_k)}
\ge-\frac{\alpha^2}{3}\|f_k\|^2_{L^2(\Omega_k)}
\end{split}
\end{equation}
for $k=1,2,3$. Since 
\begin{equation*}
\fra_{\delta,\alpha}[f] = \sum_{k=1}^3\|\nabla f_k\|^2_{L^2(\Omega_k;\dC^2)} -\frac{\alpha}{2}\sum_{k=1}^3
\|f_k|_{\partial\Omega_k}\|^2_{L^2(\partial\Omega_k)}
\end{equation*}
and
$\sum_{k=1}^3\|f_k\|_{L^2(\Omega_k)}^2 = 1$
we conclude from \eqref{123wedge} that
\[
\fra_{\delta,\alpha}[f] \ge -\frac{\alpha^2}{3},
\]
and hence  $\min\sigma(-\Delta_{\delta,\alpha}) \geq - \frac{\alpha^2}{3}$. 
Furthermore, according to Lemma~\ref{lem:wedgetrace1} there exist $f_k\in H^1(\Omega_k)$ such that equality holds in \eqref{123wedge}. 
This yields
\[
\min\sigma(-\Delta_{\delta,\alpha}) = -\frac{\alpha^2}{3}
\] 
and completes the proof of \eqref{mindelta}. We remark that the estimate $\min\sigma(-\Delta_{\delta,\alpha}) \leq -\frac{\alpha^2}{3}$
follows also from  \cite[Theorem 3.2]{BEW09}.

\subsection{Proof of \eqref{mindeltaprime}}

The proof of the estimate 
\begin{equation*}
\min(-\Delta_{\delta^\prime,\beta}) \geq - \Bigg(\frac{12\sqrt{3}-2}{9}\Bigg)^2\frac{1}{\beta^2}
\end{equation*}
is carried out in three steps, followed by a separate proof of the inequality \eqref{abc} below. 

\vskip 0.15cm
\noindent{\bf Step I.} Recall first that $-\Delta_{\delta',\beta}$ corresponds to the quadratic form
\[
\begin{split}
\fra_{\delta',\beta}[f] &=\sum_{k=1}^3\|\nabla f_k\|_{L^2(\Omega_k;\dC^2)}^2 -\beta^{-1}
\|f_1|_{\Sigma_{12}} - f_{2}|_{\Sigma_{12}}\|_{L^2(\Sigma_{12})}^2  \\
&\qquad\qquad -\beta^{-1}\|f_2|_{\Sigma_{23}} - f_{3}|_{\Sigma_{23}}\|_{L^2(\Sigma_{23})}^2-\beta^{-1}
\|f_3|_{\Sigma_{13}} - f_{1}|_{\Sigma_{13}}\|_{L^2(\Sigma_{13})}^2,\\
\dom\fra_{\delta',\beta} &=\bigoplus_{k=1}^3 H^1(\Omega_k),
\end{split}
\]
where $f_k = f|_{\Omega_k}$, $k=  1,2,3$. We split the problem into two separate problems for odd and even
components. For $f\in\dom\fra_{\delta',\beta}$ with $\|f\|_{L^2(\dR^2)} = 1$ we decompose the restrictions $f_k = f|_{\Omega_k}$
as in Lemma~\ref{lem:decomposition}.
Let $\{f_{k,\rm e}\}_{k=1}^3$ and $\{f_{k,\rm o}\}_{k=1}^3$ be the corresponding even and odd components, and let
\begin{equation*}
\begin{split}
 \theta_1 := f_{1,\rm e}|_{\Sigma_{12}}=f_{1,\rm e}|_{\Sigma_{13}},&\qquad \qquad\eta_1:=f_{1,\rm o}|_{\Sigma_{12}}=f_{1,\rm o}|_{\Sigma_{13}},\\
 \theta_2 := f_{2,\rm e}|_{\Sigma_{23}}=f_{1,\rm e}|_{\Sigma_{12}},&\qquad \qquad\eta_2:=f_{1,\rm o}|_{\Sigma_{23}}=f_{1,\rm o}|_{\Sigma_{12}},\\
 \theta_3 := f_{3,\rm e}|_{\Sigma_{13}}=f_{1,\rm e}|_{\Sigma_{23}},&\qquad \qquad\eta_3:=f_{1,\rm o}|_{\Sigma_{13}}=f_{1,\rm o}|_{\Sigma_{23}}.
\end{split}
\end{equation*}
Using Lemma~\ref{lem:decomposition}
we obtain
\[
\fra_{\delta',\beta}[f] = \sum_{k=1}^3 \|\nabla f_{k,\rm e}\|^2_{L^2(\Omega_k;\dC^2)}
+ \sum_{k=1}^3 \|\nabla f_{k,\rm o}\|^2_{L^2(\Omega_k;\dC^2)} - \frac{1}{\beta}\,S 
\]
where the value $S$ is given by 
$$
\big\|\theta_1 - \theta_2 + \eta_1 +\eta_2\big\|^2_{L^2(\dR_+)}
+\big\|\theta_2 - \theta_3 + \eta_2 +\eta_3\big\|^2_{L^2(\dR_+)}+
\big\|\theta_3 - \theta_1 + \eta_3 +\eta_1\big\|^2_{L^2(\dR_+)};
$$
here we have identified $L^2(\Sigma_{ij})=L^2(\dR_+)$, $i,j=1,2,3$. We shall show later that the above term can be estimated by 
\begin{equation}\label{abc}
S \le \big(4 - \omega(1 -t)\big) \sum_{k=1}^3\|\theta_k\|^2_{L^2(\dR_+)} + \big(4 + 3\omega t^{-1}\big)\sum_{k=1}^3\|\eta_k\|^2_{L^2(\dR_+)}
\end{equation}
for all $t > 0$ and all $\omega\in[0,1]$. With the help of this inequality we find
\begin{equation}
\label{CC}
\fra_{\delta',\beta}[f] \ge C_{\rm e} + C_{\rm o},
\end{equation}
where $C_{\rm e}$ and $C_{\rm o}$ are given by
\[
C_{\rm e}:=  \sum_{k=1}^3 \|\nabla f_{k,\rm e}\|^2_{L^2(\Omega_k;\dC^2)} -\frac{1}{\beta}\bigl(4-\omega(1-t)\bigr)
\sum_{k=1}^3\|\theta_k\|_{L^2(\dR_+)}^2
\]
and
\[
C_{\rm o} := 
\sum_{k=1}^3 \|\nabla f_{k,\rm o}\|^2_{L^2(\Omega_k;\dC^2)} - \frac{1}{\beta}\bigl(4+3 \omega t^{-1}\bigr)\sum_{k=1}^3\|\eta_k\|^2_{L^2(\dR_+)}.
\]
\vskip 0.15cm
\noindent {\bf Step II.} In this step we estimate $C_{\rm e}$ and $C_{\rm o}$.
We start with $C_{\rm e}$.
Applying Lemma~\ref{lem:wedgetrace1} with 
$\gamma = \tfrac{1}{2\beta}(4-\omega(1-t))$ and $\varphi = 2\pi/3$ 
to the functions $\{f_{k,\rm e}\}_{k=1}^3$ we get
\[
\|\nabla f_{k,\rm e}\|^2_{L^2(\Omega_k;\dC^2)} - \frac{1}{\beta}\bigl(4-\omega(1-t)\bigr)\|\theta_k\|^2_{L^2(\dR_+)} \ge -\frac{1}{3\beta^2}
\bigl(4-\omega(1-t)\bigr)^2\|f_{k,\rm e}\|^2_{L^2(\Omega_k)}
\]
for $k=1,2,3$.
Summing up these three inequalities we find
\begin{equation}
\label{Ce}
C_{\rm e} \ge -\frac{1}{3\beta^2}
\bigl(4-\omega(1-t)\bigr)^2\sum_{k=1}^3\|f_{k,\rm e}\|^2_{L^2(\Omega_k)}
\end{equation}
for all $t > 0$ and all $\omega\in[0,1]$.
Next we estimate $C_{\rm o}$. Note that $f_{k,{\rm o}}|_{\Sigma_k} = 0$ for $k=1,2,3$, and
hence we can apply Lemma~\ref{lem:wedgetrace2} with $\gamma = \frac{1}{2\beta}(4+3 \omega t^{-1})$ to the 
functions $\{f_{k,{\rm o}}\}_{k=1}^3$. This yields 
\[
\|\nabla f_{k,\rm o}\|^2_{L^2(\Omega_k;\dC^2)} - \frac{1}{\beta}\bigl(4+3 \omega t^{-1}\bigr)\|\eta_k\|^2_{L^2(\dR_+)} \ge -
\frac{1}{4\beta^2}\bigl(4+3 \omega t^{-1}\bigr)^2\|f_{k,\rm e}\|^2_{L^2(\Omega_k)}
\]
for $k=1,2,3$. Summing up these three inequalities gives
\begin{equation}
\label{Co}
C_{\rm o} \ge -\frac{1}{4\beta^2}\bigl(4+3 \omega t^{-1}\bigr)^2\sum_{k=1}^3\|f_{k,\rm o}\|^2_{L^2(\Omega_k)}
\end{equation}
for all $t > 0$ and all $\omega\in[0,1]$.

\vskip 0.15cm
\noindent{\bf Step III.}
Note that 
\[
\sum_{k=1}^3 \big(\|f_{k,\rm e}\|_{L^2(\Omega_k)}^2 + \|f_{k,\rm o}\|_{L^2(\Omega_k)}^2\big) = \|f\|^2_{L^2(\dR^2)} = 1.
\]
Thus, \eqref{CC}, \eqref{Ce} and \eqref{Co} imply
\begin{equation}\label{mnb}
\fra_{\delta',\beta}[f] \ge - \inf_{t > 0}
\min_{\omega\in[0,1]}\max\left\{\frac{1}{3}\bigl(4-\omega(1-t)\bigr)^2,\frac{1}{4}\left(4+\frac{3 \omega}{t}\right)^2\right\}\frac{1}{\beta^2}
\end{equation}
and for \eqref{mindeltaprime} to hold it remains to show that the value on the right hand side is equal to $-(\frac{12\sqrt{3}-2}{9})^2\frac{1}{\beta^2}$.
For this consider the functions 
\begin{equation*}
M_1(\omega,t) := \frac{1}{3}\bigl(4-\omega(1-t)\bigr)^2\quad\text{and}\quad M_2(\omega,t) 
:= \frac{1}{4}\left(4+\frac{3 \omega}{t}\right)^2.
\end{equation*}
We have
\[
M_1(0,1) = \frac{16}{3},\quad M_2(0,1) = 4,\quad\text{and}\quad M_1(\omega,t) \ge \frac{16}{3},\,\,\,t \ge 1,\,\omega\in[0,1],
\]
and therefore
\begin{equation}
\label{1infty}
\min_{t\in [1,+\infty)}\min_{\omega\in[0,1]}\max\big\{M_1(\omega,t), M_2(\omega,t)\big\} = \frac{16}{3}.
\end{equation}
Suppose now that $t\in(0,1)$ is fixed. Then $M_1(\cdot,t)$ is continuous and decreasing, 
whereas $M_2(\cdot,t)$ is continuous and increasing, and a straightforward computation shows that for
\begin{equation*}
\omega_* := \frac{(8 - 4\sqrt{3})t}{3\sqrt{3} + 2(1-t)t}\in [0,1]
\end{equation*}
we have $M_1(\omega_*,t) = M_2(\omega_*,t)$. Hence for $t\in(0,1)$ fixed we find
\begin{equation*}
\begin{split}
\min_{\omega\in[0,1]}\max\big\{M_1(\omega,t), M_2(\omega,t)\big\} & = M_1(\omega_*,t) = M_2(\omega_*,t)\\
&=\frac{1}{4}\Bigg(4 + \frac{3(8 - 4\sqrt{3})}{3\sqrt{3} + 2(1-t)t}\Bigg)^2.
\end{split}
\end{equation*}
Next, we minimize with respect to $t\in(0,1)$. Clearly, the above value is minimal in $t\in(0,1)$
if $t = 1/2$. 
Therefore we obtain
\begin{equation}
\label{01}
\min_{t\in(0,1)}\min_{\omega\in[0,1]} \max\{M_1(\omega,t), M_2(\omega,t)\}= M_2(\omega_*,1/2)= \Bigg(\frac{12\sqrt{3}-2}{9}\Bigg)^2.
\end{equation}
Now \eqref{1infty} and \eqref{01} together imply 
\[
\min_{t>0}\!\min_{\omega\in[0,1]}\!\! \max\{M_1(\omega,t), M_2(\omega,t)\} \!=\!  \min\Bigg\{\frac{16}{3},\Bigg(\frac{12\sqrt{3}-2}{9}\Bigg)^2\Bigg\}
\!=\!
\Bigg(\frac{12\sqrt{3}-2}{9}\Bigg)^2\!,
\]
and hence the assertion \eqref{mindeltaprime} on the minimum of the spectrum of $-\Delta_{\delta',\beta}$ follows from \eqref{mnb}.

\begin{proof}[Proof of the estimate \eqref{abc}]
Let $\theta_1,\theta_2,\theta_3\in L^2(\dR_+) $ and  $\eta_1, \eta_2, \eta_3\in L^2(\dR_+)$. We shall not use an index for the norm in $L^2(\dR_+)$
in this proof. From
\[
\begin{split}
\|\theta_1 - \theta_2 + \eta_1 + \eta_2\|^2 &\le 2\|\theta_1 +\eta_1\|^2 + 2\|\theta_2 - \eta_2\|^2,\\
\|\theta_2- \theta_3 + \eta_2 + \eta_3 \|^2 &\le 2\|\theta_2 +\eta_2\|^2 + 2\|\theta_3 - \eta_3\|^2,\\
\|\theta_3- \theta_1 + \eta_3 + \eta_1\|^2 &\le 2\|\theta_3 +\eta_3\|^2 + 2\|\theta_1 - \eta_1\|^2,\\
\end{split}
\]
and the parallelogram identity we conclude
\begin{equation}\label{Estimate1}
S \le 2\sum_{k=1}^3\big(\|\theta_k - \eta_k\|^2 + \|\theta_k + \eta_k\|^2\big)\le 4\sum_{k=1}^3 \big(\|\theta_k\|^2 + \|\eta_k\|^2\big).
\end{equation}

On the other hand we have
\begin{equation}\label{S}
\begin{split}
S &= \|\theta_1 - \theta_2\|^2 + \|\theta_2 - \theta_3\|^2 + \|\theta_3 - \theta_1\|^2 + \|\eta_1 + \eta_2\|^2 + \|\eta_2+\eta_3\|^2 + \|\eta_3 + \eta_1\|^2\\[0.2ex]
&\quad+2\Real\big[\big(\theta_1-\theta_2, \eta_1 +\eta_2\big)\big] + 2\Real\big[\big(\theta_2-\theta_3,\eta_2 +\eta_3\big)\big] + 2\Real\big[\big(\theta_3  -\theta_1,\eta_3 +\eta_1\big)\big]\\
&= \|\theta_1 - \theta_2\|^2 + \|\theta_2 - \theta_3\|^2 + \|\theta_3 - \theta_1\|^2 + \|\eta_1 + \eta_2\|^2 + \|\eta_2+\eta_3\|^2 + \|\eta_3 + \eta_1\|^2\\[0.2ex]
&\quad+2\Real\big[\big(\theta_1, \eta_2 -\eta_3\big)\big] + 
2\Real\big[\big(\theta_2,\eta_3 -\eta_1\big)\big] + 2\Real\big[\big(\theta_3,\eta_1 -\eta_2\big)\big]
\end{split}
\end{equation}
and the Cauchy-Schwarz inequality together with the inequality $2ab \le ta^2 + \tfrac{1}{t}b^2$, $a,b >0$, $t >0$, 
yields
\[
\begin{split}
\Big|2\Real\big[\big(\theta_1,\eta_2 -\eta_3\big)\big]\Big| &\le 2 \|\theta_1\|\cdot\|\eta_2 - \eta_3\|\le t \|\theta_1\|^2 + \tfrac{1}{t}\|\eta_2 - \eta_3\|^2,\\[0.2ex]
\Big|2\Real\big[\big(\theta_2,\eta_3 -\eta_1\big)\big]\Big| &\le 2 \|\theta_2\|\cdot\|\eta_3 - \eta_1\|\le t\|\theta_2\|^2 +\tfrac{1}{t}\|\eta_3 - \eta_1\|^2,\\[0.2ex]
\Big|2\Real\big[\big(\theta_3,\eta_1 -\eta_2\big)\big]\Big| &\le 2 \|\theta_3\|\cdot\|\eta_1 - \eta_2\|\le t   \|\theta_3\|^2 + \tfrac{1}{t}\|\eta_1 - \eta_2\|^2.
\end{split}
\]
Combining the latter with \eqref{S} and making use of 
$\Vert\eta_i+\eta_j\Vert^2\leq 2\Vert\eta_i\Vert^2 + 2\Vert\eta_j\Vert^2$, $i,j=1,2,3$, we arrive at
\begin{equation}
\label{ineqt}
\begin{split}
S &\le \|\theta_1 - \theta_2\|^2 + \|\theta_2 - \theta_3\|^2  + \|\theta_3 - \theta_1\|^2 + 
4\|\eta_1\|^2 + 4\|\eta_2\|^2 + 4\|\eta_3\|^2  \\
& +t\bigl(\|\theta_1\|^2 + \|\theta_2\|^2 + \|\theta_3\|^2\bigr) + \tfrac{1}{t}\bigl(\|\eta_1- \eta_2\|^2 + 
\|\eta_2 -\eta_3\|^2 + \|\eta_3-\eta_1\|^2\bigr)
\end{split}
\end{equation}
for all $t > 0$. Moreover, as 
\begin{equation*}
\begin{split}
\|\theta_1 - \theta_2\|^2 + \|\theta_2 - \theta_3\|^2 + \|\theta_3 - \theta_1\|^2 
& = 3\|\theta_1\|^2 + 3\|\theta_2|^2 + 3\|\theta_3\|^2 - \|\theta_1 + \theta_2 + \theta_3\|^2\\
&\le 3\|\theta_1\|^2 + 3\|\theta_2\|^2 + 3\|\theta_3\|^2,
\end{split}
\end{equation*}
and analogously
\begin{equation*}
\|\eta_1 - \eta_2\|^2 + \|\eta_2 - \eta_3\|^2 + \|\eta_3 - \eta_1\|^2 \le 3\|\eta_1\|^2 + 3\|\eta_2\|^2 + 3\|\eta_3\|^2,
\end{equation*}
we can further estimate \eqref{ineqt} by
\begin{equation}
\label{Estimate2}
S\le \big(3+t\big)\sum_{k=1}^3\|\theta_k\|^2 + \sum_{k=1}^3 \big(4+ \tfrac{3}{t}\big)\|\eta_k\|^2.
\end{equation}
In order to obtain the estimate \eqref{abc} let
$\omega \in [0,1]$, consider $S = \omega S + (1 - \omega) S$ and 
estimate $\omega S$ as in \eqref{Estimate2} and $(1 - \omega)S$ as in \eqref{Estimate1}. 
\end{proof}


\end{document}